\documentclass[aps, prl, reprint, twocolumn, superscriptaddress,
    longbibliography, floatfix, showkeys,nofootinbib]{revtex4-2}

\usepackage{amsthm,amsmath,mathtools,latexsym,amssymb,verbatim,enumerate}

\usepackage[table,xcdraw]{xcolor}

\usepackage{color}
\usepackage{dsfont}
\usepackage{subeqnarray}
\usepackage{lettrine}

\usepackage{newpxtext,newpxmath}
\usepackage{moresize}
\usepackage{mathtools,mathrsfs,bbding}
\usepackage{microtype}
\usepackage{chngcntr}
\usepackage{siunitx}
\usepackage{braket}
\usepackage{geometry}
\usepackage{graphicx}

\usepackage{hyperref}
\usepackage{cleveref}
\usepackage{appendix}

\DeclareMathAlphabet{\mathcal}{OMS}{cmsy}{m}{n}

\usepackage[T1]{fontenc}
\usepackage[utf8]{inputenc}
\usepackage{adjustbox}

\definecolor{darkblue}{rgb}{0.0, 0.0, 0.5}
\definecolor{midnightblue}{rgb}{0.1378,0.3784,0.4838}
\definecolor{grayblue}{rgb}{0.0, 0.3, 0.5}
\definecolor{graypink}{rgb}{0.75, 0.35, 0.25}
\definecolor{darkgreen}{rgb}{0.0, 0.25, 0.0}
\definecolor{darkpurple}{rgb}{0.3, 0.0, 0.25}
\definecolor{darkorange}{rgb}{1.0, 0.549, 0.0}
\definecolor{darkred}{rgb}{0.5, 0.0, 0.0}
\definecolor{verylightgray}{rgb}{0.9, 0.9, 0.9}
\definecolor{gray75}{gray}{0.3}

\hypersetup{colorlinks=true, citecolor=grayblue, linkcolor=grayblue, urlcolor=grayblue}

\newcommand{\degree}{$^{\circ}$ }

\makeatletter
\newcommand{\shorteq}{%
  \settowidth{\@tempdima}{\,--}% Width of hyphen
  \resizebox{\@tempdima}{\height}{\,=}%
}

\newcommand{\tr}{\mathrm{Tr}\hspace{0.05cm}}
\newcommand{\ii}{\mathrm{i}}
\newcommand{\ee}{\mathrm{e}}
\newcommand{\dop}{\mathrm{do}}

\newcommand{\pos}{\succcurlyeq}

\usepackage{lipsum}

\theoremstyle{plain}																		% Theorem, Lemma, Proposition and Corollary

\newtheorem{thm}{Theorem}
\newtheorem{lem}[thm]{\bf Lemma}
\newtheorem{prop}[thm]{Proposition}

\renewcommand{\thefootnote}{\fnsymbol{footnote}}

\newcommand\blfootnote[1]{%
	\begingroup
	\renewcommand\thefootnote{}\footnote{#1}%
	\addtocounter{footnote}{-1}%
	\endgroup
}

\begin{document}

\title{Device-independent quantum memory certification in two-point measurement experiments}

\author{Leonardo S. V. Santos$^\ast$}
\email{leonardo.svsantos@student.uni-siegen.de}
\affiliation{Naturwissenschaftlich-Technische Fakultät, Universität Siegen, Walter-Flex-Straße 3, 57068 Siegen, Germany}

\author{Peter Tirler$^\ast$}
\email{Peter.Tirler@uibk.ac.at}
\affiliation{Universität Innsbruck, Institut für Experimentalphysik, Technikerstraße 25, 6020 Innsbruck, Austria}

\author{Michael Meth}
\affiliation{Universität Innsbruck, Institut für Experimentalphysik, Technikerstraße 25, 6020 Innsbruck, Austria}
\author{Lukas Gerster}
\affiliation{Universität Innsbruck, Institut für Experimentalphysik, Technikerstraße 25, 6020 Innsbruck, Austria}
\author{Manuel John}
\affiliation{Universität Innsbruck, Institut für Experimentalphysik, Technikerstraße 25, 6020 Innsbruck, Austria}
\author{Keshav Pareek}
\affiliation{Universität Innsbruck, Institut für Experimentalphysik, Technikerstraße 25, 6020 Innsbruck, Austria}
\author{Tim Gollerthan}
\affiliation{Universität Innsbruck, Institut für Experimentalphysik, Technikerstraße 25, 6020 Innsbruck, Austria}

\author{Martin Ringbauer}
\email{martin.ringbauer@uibk.ac.at}
\affiliation{Universität Innsbruck, Institut für Experimentalphysik, Technikerstraße 25, 6020 Innsbruck, Austria}

\author{Otfried G\"{u}hne}
\email{otfried.guehne@uni-siegen.de}
\affiliation{Naturwissenschaftlich-Technische Fakultät, Universität Siegen, Walter-Flex-Straße 3, 57068 Siegen, Germany}
\blfootnote{$^\ast$ These authors contributed equally}

\date{\today}     

\begin{abstract}
Quantum memories are key components of emerging quantum technologies. They are designed to store quantum states and retrieve them on demand without losing features such as superposition and entanglement. Verifying that a memory preserves these features is indispensable for applications such as quantum computation, cryptography and networks, yet no general and assumption-free method has been available. Here, we present a device-independent approach for certifying black-box quantum memories, requiring no trust in any part of the experimental setup. We do so by probing quantum systems at two points in time and then confronting the observed temporal correlations against classical causal models through violations of causal inequalities. We perform a proof-of-principle experiment in a trapped-ion quantum processor, where we certify \SI{35}{\milli \second} of a qubit memory. Our method establishes temporal correlations and causal modelling as practical and powerful tool for benchmarking key ingredients of quantum technologies, such as quantum gates or implementations of algorithms.
% https://www.nature.com/documents/nature-summary-paragraph.pdf
\end{abstract}

\maketitle

\section{Introduction}
A memory is a device that reliably stores information. In the classical setting, a familiar example is a hard drive, which retains bits that can be retrieved on demand. In quantum information processing, however, the task is considerably more demanding: a memory must store quantum states, preserving not only any classical information they encode but also intrinsically quantum features such as superposition and entanglement. These properties are exquisitely fragile and easily disturbed by noise, yet they are essential for achieving quantum advantages. A quantum memory must therefore safeguard them throughout the storage time to ensure the reliable performance of quantum computation~\cite{DKLP02,Nielsen03,Terhal15}, the implementation of quantum network architectures~\cite{ABCEL21,LKZJS23}, and secure quantum cryptography~\cite{PR22}. %Bravyi2024 Kimble08,WEH18,PABBBCEGLOPRSTUVVW20,RW23

A key step towards the practical usage of quantum memories is their certification; that is, rigorously confirming their correct functioning. This is particularly pressing for quantum key distribution, where noise or loss can reduce key rates or open side channels for eavesdropping~\cite{SBCDLP09}. Entanglement-based quantum networks likewise depend on high-fidelity memories to buffer entangled states during entanglement swapping; here, noise limits the success probability and hence the scalability of long-distance entanglement distribution~\cite{AEEHJLT23}.

\begin{figure}
    \centering
    \includegraphics[width=1\linewidth]{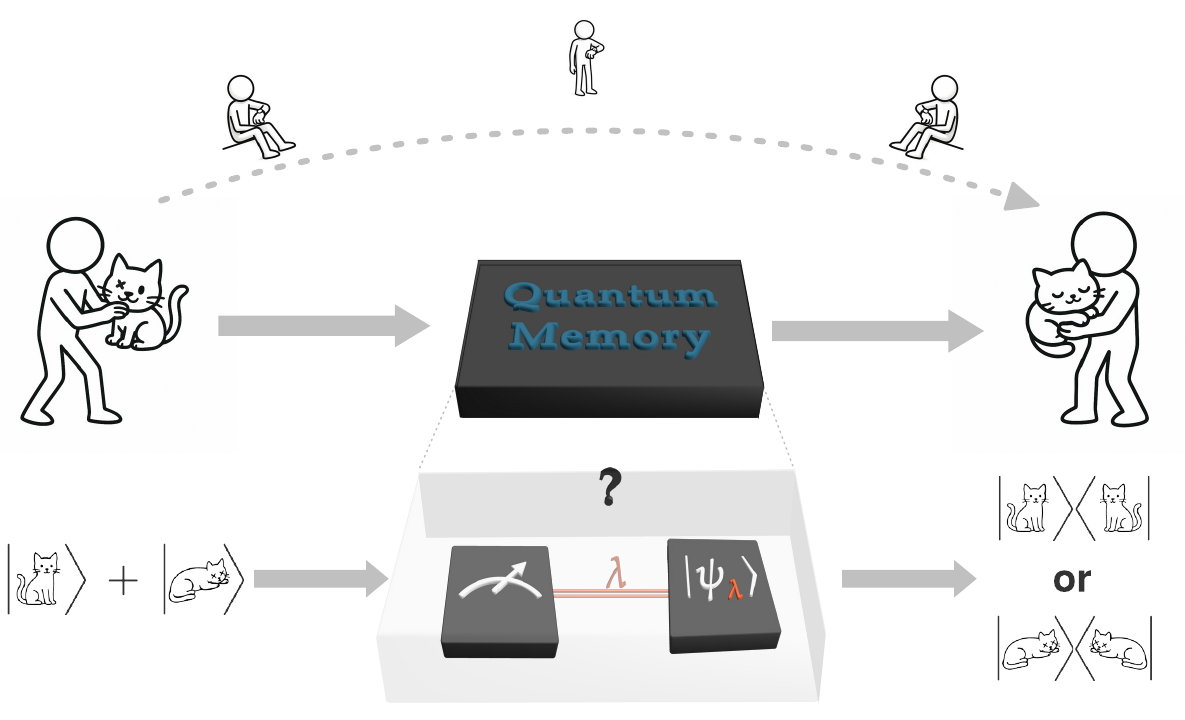}
    \caption{\sf \textbf{Quantum memory certification.} A quantum memory is expected to reliably store quantum states, yet real devices are imperfect. Noise can induce decoherence and dissipation, effectively classicalising the memory so that only classical information is transmitted. If that happens, the process is reduced to one in which the system is measured and a state is prepared based on the measurement outcome. In other words, the device no longer functions as a genuine quantum memory and must be ruled out.}
    \label{fig:catmemory}
\end{figure}

\begin{figure*}[t]\flushleft
    \includegraphics[width=1\linewidth]{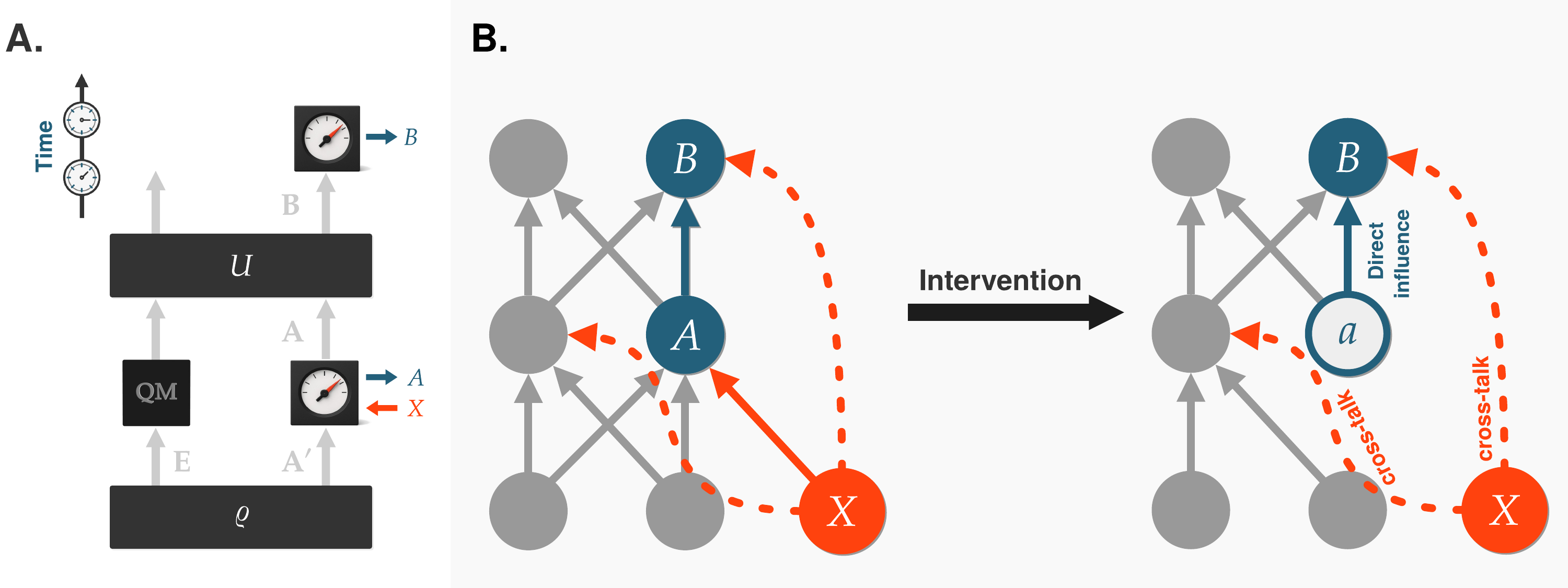}
    \caption{\sf \textbf{Two-point measurement experiments.} {\bf\textsf{A.}} To test a quantum memory (QM), we consider the scenario where quantum systems are measured at two consecutive points in time. {\bf\textsf{B.}} The observed statistics are then compared against a classical causal model in which outcomes are part of a time series. These models are represented by Bayesian networks, where the nodes correspond to the variables which are connected by arrows if causal influence is possible. Interventions make causal influences explicit and empirically accessible. Classical-quantum gaps in temporal correlations certify the quantum memory.}
    \label{fig:cover}
\end{figure*}

When it comes to certification, the device-independent (DI) framework sets the gold standard: it draws conclusions solely from the observed correlations, requiring no assumptions about a device’s internal workings or the calibration of its measurements~\cite{Bancal14}. This is particularly pertinent for memories, whose certification depends on other components, such as state-preparation sources and detectors, whose own characterisation may, in turn, rely on yet further devices. DI methods offer an elegant resolution to this problem by enabling certification that is entirely agnostic to the internal details of any specific apparatus. This is possible because, under a given causal structure, classical and quantum theories can produce incompatible predictions for observable correlations. For instance, in a Bell test, subsystems of a multipartite quantum system are distributed to spatially separated parties; a violation of a Bell inequality then certifies the underlying quantum state~\cite{KR21}. % SB20

Recent years have witnessed substantial progress in certifying quantum memories across a spectrum of trust assumptions~\cite{HL09,MR13,Pusey15,FSTGLBBAG17,RBL18,YLZRTG21}, including significant steps toward DI~\cite{SBWS18,GPBPMKRRF20,YSZBFKALZXBP21,CKZTC21,SBIAB23}. Nevertheless, most DI protocols rely on Bell tests, a scenario that is not well aligned with memory certification for at least two reasons. First, memory usage is inherently temporal: information is encoded, stored, and later retrieved, so the relevant correlations occur across time rather than space. Second, memories are local modules, operating within network nodes or as components of quantum processors, so a local mode of certification is both natural in principle and desirable in practice.

Here we solve this problem by introducing a DI certification scheme based on two-point measurement (TPM) experiments, in which quantum systems are measured at two successive times. Unlike a Bell test, this scheme is expressly designed for the inherently temporal nature of memory usage. The key idea is to confront the observed temporal correlations of measurement outcomes with those that could arise from a classical stochastic process. Specifically, introducing an instrumental variable to interrogate the process allows to access causal influences and exposes classical–quantum gaps, thereby certifying the genuine quantum nature of the memory. As a further advantage, the same framework also certifies other significant features of quantum causal mechanisms, measurement incompatibility and quantum non-Markovianity.

Our framework is directly applicable to widely available NISQ-era quantum platforms, and we demonstrate it through proof-of-principle experiments on a trapped-ion quantum processor. This platform is particularly well suited for TPM-based certification because it enables mid-circuit measurement and re-preparation. In addition, trapped-ion measurements do not suffer from data loss due to no-detection events in photodetectors, thereby avoiding the need for detection-efficiency assumptions such as fair sampling that are often required in photonics-based tests.

\section{Quantum two-point measurement experiments}

Consider a scenario in which quantum systems are probed at two points in time as shown in Fig.~\ref{fig:cover}. Initially, the system $\bf A^\prime$ may be correlated with another system $\bf E$, described by the joint state $\varrho$. A measurement is performed on $\bf A^\prime$, yielding a classical outcome $A$, and leaving behind a post-measurement system $\bf A$, which may differ from $\bf A^\prime$. This system then evolves in time, jointly with $\bf E$, under a global unitary transformation $U$. At a later time, a second measurement is performed on a system $\bf B$, yielding a classical outcome $B$. This scenario shall be used to test whether the system $\bf E$ acts as a quantum memory.

An \emph{ideal} quantum memory must faithfully preserve quantum states, which means it is described by a reversible (unitary) channel. Any degradation of the memory leads to loss of information, introducing irreversibility and driving the stored state towards classicality. Once a memory becomes fully classical, it can transmit \emph{at most} classical information. 
Formally, this corresponds to the input undergoing a measurement whose outcome is the only information retained (see Fig~\ref{fig:catmemory}); equivalently, the memory's channel becomes {entanglement-breaking}~\cite{HSR03}: it destroys any entanglement the input system may have had with other degrees of freedom.

Our goal is to certify a black-box quantum memory using only the observed statistical correlations. By certification we mean both ruling out any classical simulation and assessing its quality---namely, how far it departs from an ideal memory. We quantify this ``quality’’ using the Uhlmann fidelity between the Choi–Jamio{\l}kowski state~\cite{Choi75,J72} of the implemented memory channel and that of an ideal one, which we denote by $F_{\rm Memory}$.

For defining non-classicality, we follow the spirit of Bell's theorem~\cite{Bell64}, consider statistical correlations and ask whether they can be explained within a classical model. In our setting, the variables $A$ and $B$ represent outcomes of sequential measurements, suggesting a natural causal ordering $A\to B$; that is, we assume a forward-directed time series in which $A$ may influence $B$, but not the reverse. These variables may also interact with an additional stochastic process that is potentially inaccessible (latent), which plays a role analogous to that of \textit{hidden variables} in Bell tests. For analysing the behaviour of a system, we also consider variables that the experimenter samples independently of the systems being interrogated. Here, we focus on a single such variable, denoted $X$, which is used to determine the measurement setting at the first time step.

The causal assumptions of the model are summarised in the panel {\bf\textsf{B}} of Fig~\ref{fig:cover}. In particular, notice that $X$ may affect $B$ through a direct causal link, indirectly via its effect on $A$ or through the hidden process. We refer to causal influences of $X$ on $B$ that are \emph{not} mediated through $A$ as {\emph{cross-talk influences}}. Yet, we assume that no variable that influences $A$ or $B$ also influences $X$, an assumption that is reminiscent of the assumption of \textit{free will} or \textit{measurement independence} in Bell experiments\footnote{Note that $X$ is sampled from an external source of randomness, such as a pseudo-random number generator, human choices, or even cosmic signals from distant quasars. These sources have no reason whatsoever to be causally connected with the system under investigation.}.

To determine whether these classical models can account for quantum TPM correlations we proceed as follows. First, we adopt the framework of potential outcome models from causal inference~\cite{Rubin74} and consider an idealised scenario with no cross-talk influence. In this case, the scenario can be mapped onto the well-known {\emph{instrumental network}}~\cite{Pearl95,BP97,Bonet01,KM20}. We derive an inequality, whose violation certifies that the correlations cannot be described by such a model, implying that the quantum memory cannot be described by a classical mechanism. We then go beyond this idealised setting by exploring deviations, drawing on insights from the theory of quantum contextuality with disturbing measurements~\cite{KDL15}. Finally, we obtain testable implications in the form of Bell-like inequalities that allows us to certify whether TPM can be explained classically. 

Specifically, we consider $A$ and $B$ as binary, and $X$ as having finite alphabet $\mathcal{X}$. First, we define
\begin{align}\label{eq:km}
\Gamma_{\Pr(AB\mid X)}&=\sum_{b_0,\,b_1}\inf_{x\in \mathcal{X}}\sum_{a}\Pr(A=a,B=b_a\mid X=x).
\end{align}
Then, as shown in the \hyperlink{SM}{Supplementary Information}, for a classical
model the inequality $\Gamma_{\Pr(AB\mid X)} \geq 1$ holds.

Next, to account for cross-talk influences, we consider \emph{interventions} upon $A$ to somehow \textit{force it to assume a value} predefined \emph{a priori}~\cite{Pearl09}. This is done in order to stop the causal flux from $X$ to $B$ through the variable $A$, so that any remaining influence $X\to B$ is \emph{exclusively} due to cross-talk. We denote $B_{\dop(A=a,X=x)}$ as the variable representing the outcome of the second measurement when $A$ is set to a value $a$ and $X$ to $x$. Cross-talk influence can then be quantified in terms of the \textit{average causal direct effect} (ACDE),
\begin{align}\label{eq:ACDE}
\sup|\Pr(B_{\dop(A=a,X=x)}=b)-\Pr(B_{\dop(A=a,X=x^\prime)}=b)|,
\end{align}
the supremum being taken over all values indicated. Assuming classicality then leads to the following inequality
\begin{equation}\label{eq:kmct}
\Gamma_{\Pr(AB\mid X)}+2\mathrm{ACDE}\geq 1.
\end{equation}
This inequality depends on two different experiments (with and without intervention) so it does \emph{not} have the same status as a Bell inequality. As such, we treat cross-talk influence as a \emph{loophole} that must be avoided and addressed. Without interventions, crosstalk can be witnessed via {\emph{Pearl's inequality}}~\cite{Pearl95}, defined in terms of
\begin{equation}\label{eq:Pearl}
\Delta_{\Pr(AB\mid X)}=\max_{a\in\{0,1\}}\sum_{b}\sup_{x\in \mathcal{X}}\Pr(A=a,B=b\mid X=x),
\end{equation}
whence absence of cross-talk influence implies $\Delta_{\Pr(AB\mid X)}\leq 1$. Interestingly, Pearl's inequality is a pre-condition for absence of cross-talk influences regardless of the physical theory~\cite{HLP14}; in other words, its violation signals cross-talk rather than any form of non-classicality.

Inequality~\eqref{eq:kmct}, in turn, can be violated by quantum TPM correlations. In particular, if there is no cross-talk influence (i.e., $\mathrm{ACDE}=0$), it can be violated down to
\begin{equation}\label{eq:max-violation}
\Gamma_{\Pr(AB\mid X)}^{\,\star}=\underbrace{2-\sqrt{2}}_{\approx \, 0.586}<1,
\end{equation}
a quantum bound that is tight. Finally, the observed violation of classicality can be translated into a DI lower bound on the memory's quality. In terms of $\Gamma_{\Pr(AB\mid X)}$ [Eq.~\eqref{eq:km}], $F_{\rm Memory}$ for a qubit memory can be a lower bounded by
\begin{equation}\label{eq:fidelity}
F_{\rm Memory}\geq\frac{1}{2}\left(1-\frac{\Gamma_{\Pr(AB\mid X)}-2+S_{\rm K}}{\sqrt{2}-S_{\rm K}}\right),
\end{equation}
where $S_{\rm K}=(8+7\sqrt{2})/17$ (see \hyperlink{SM}{Supplemental Information} for details).

\begin{figure*}[ht]
    \begin{minipage}[t]{0.53\textwidth}\flushleft
    {\bf \textsf{A.}}
    \includegraphics[width=1\linewidth]{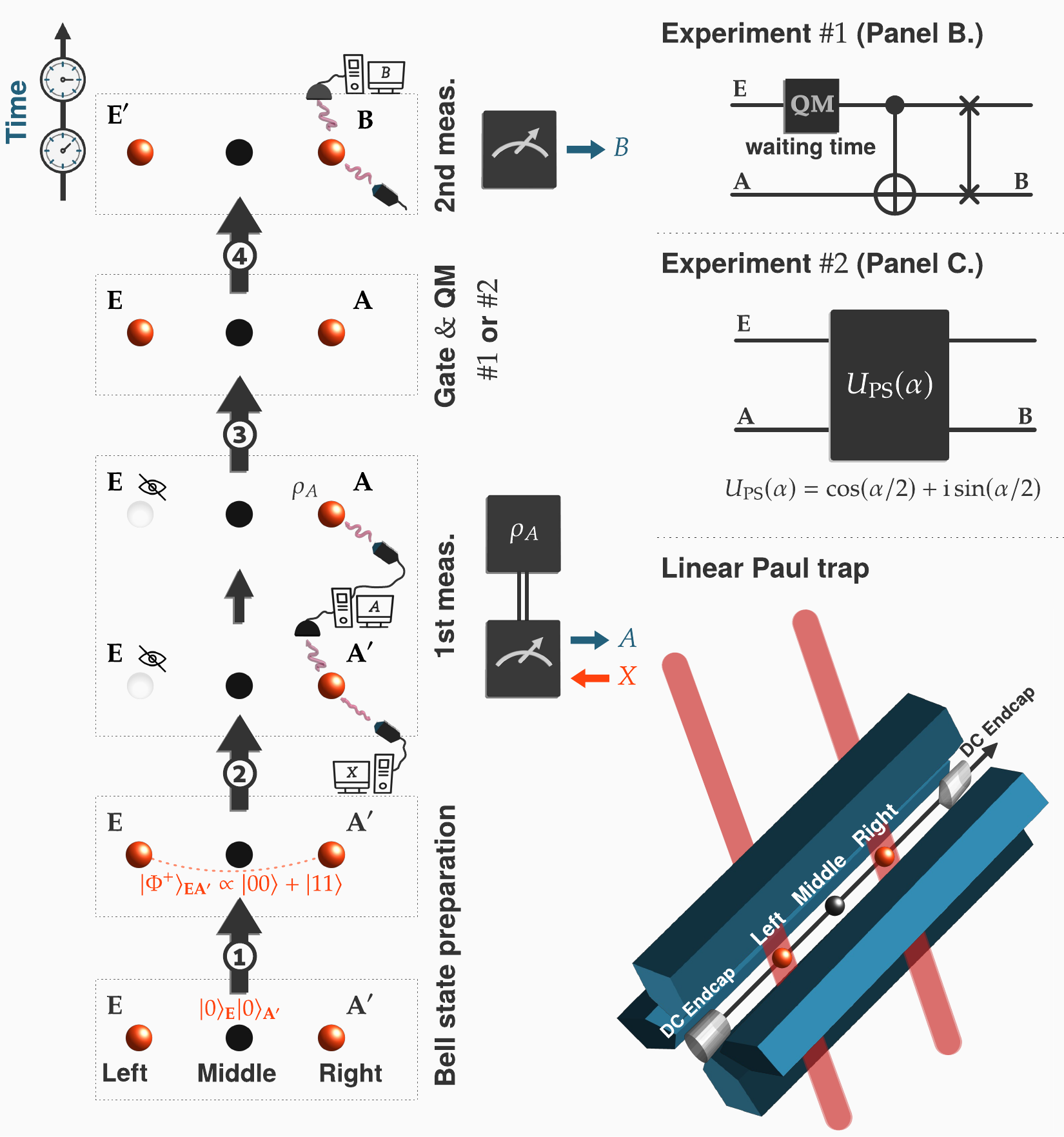}
    \vspace{0.425cm}
    \end{minipage}
    \begin{minipage}[t]{0.44\textwidth}\flushleft
    {\bf \textsf{B.}}\\
    \includegraphics[width=1\linewidth]{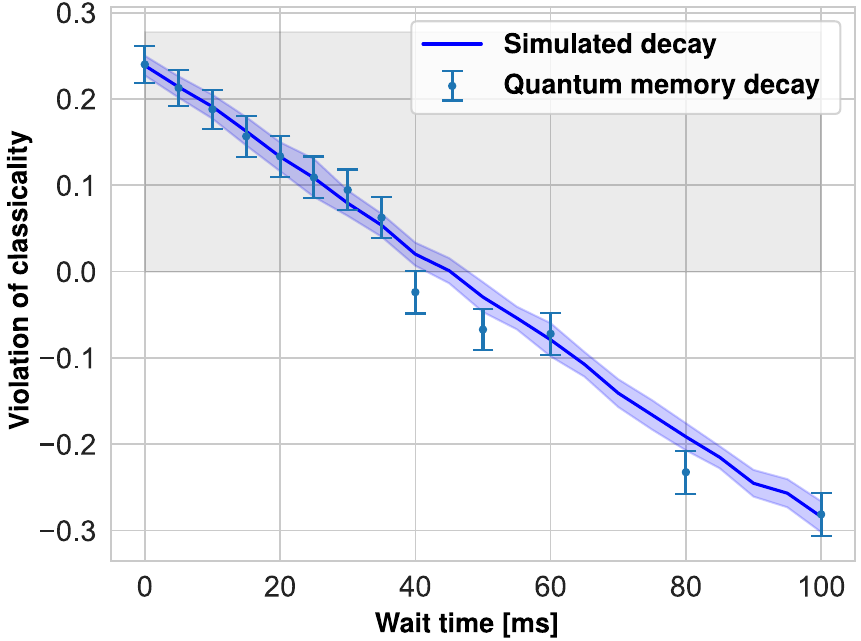}\\
    {\bf \textsf{C.}}\\ 
    \includegraphics[width=1\linewidth]{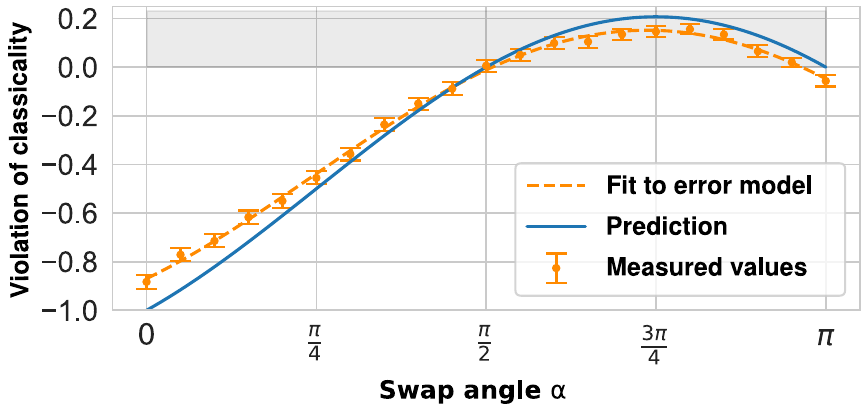}
    \end{minipage}
    \caption{\small\sf \textbf{Experiment and results.} 
    Panel {\bf \textsf{A.}} shows the implemented sequence step by step. The ions are confined in a macroscopic linear Paul trap, and we consider two experiments ($\#1$ and $\#2$). Panels {\bf \textsf{B.}} and {\bf \textsf{C.}} display the experimental results, rescaled such that positive values indicate a violation of classicality. In {\bf \textsf{B.}}, we observe the classicalisation of a quantum memory with increasing waiting time, where the dark-blue line denotes the expected behaviour from dephasing and fidelity, with the blue shaded region indicating one-$\sigma$ shot-noise uncertainty. In {\bf \textsf{C.}}, quantumness is shown versus the partial $\alpha$-\textsc{swap}; violations for $\pi/2 < \alpha < \pi$ certify quantum causal relations, excluding classical causal models. The blue curve represents the ideal prediction.}
    \label{fig:experiment}
\end{figure*}

\section{Certifying a quantum memory}

We proceed with the experimental implementation. The setup consists of a chain of three $^{40}\text{Ca}^+$ ions confined in a linear Paul trap; see Fig.~\ref{fig:experiment} and also Refs.~\cite{Schindler2013,Ringbauer2022}. The two outer ions encode the qubits $\bf A^\prime$ and $\bf E$. The middle ion is used solely for spatial separation and does not actively participate in the protocol. The pair $\bf A^\prime+\bf E$ is initialised in a Bell state via a single application of a M{\o}lmer-S{\o}rensen gate~\cite{molmer_sorensen_1999}. Experimental details and additional data are provided in the \hyperlink{SM}{Supplementary Information}

After the initial state preparation, qubit $\bf A^\prime$ is measured: the variable $X$ is sampled to select one of four possible measurements, 
$\pm \sigma_x$- and $\pm \sigma_z$- Pauli observables, the outcome of which corresponds to the event $A=a$. We implement a \emph{measure-and-prepare instrument} in which the post-measurement system $\bf A$ is prepared in a state $\rho_A$ that depends only on the observed outcome, and not on the chosen measurement setting. Specifically, for each $x$, we measure $\bf A^\prime$ with a POVM $E_{A\mid X=x}$, the result of which determines a post-measurement state $\rho_{A=a}$ for qubit $\bf A$ conditioned on the event $A=a$. It is important that this state depends only on the observed event $A=a$ and not on the sampled value $X=x$. To prevent any effect of the measurement on qubit $\bf E$, the corresponding ion is spectroscopically decoupled by shelving it in an unused energy level during the measurement. While this technique is very effective, perfect experimental isolation of the two qubits is always an idealisation; deviations from it must be addressed with the ACDE and Pearl's inequality.

After the first measurement, the pair of qubits $\bf A+E$ evolves through a unitary $U$, becoming qubits $\bf B + E^\prime$. Afterwards, qubit $\bf B$ is measured\footnote{As we treat systems at different times as \emph{potentially distinct}, qubit $\bf B$ may correspond to either the left or the right ion (cf.~Fig.~\ref{fig:experiment}), depending on the protocol.}. The unitary used to certify the quantum memory consists of a {\sc cnot} gate followed by a unitary {\sc swap}. The second measurement on qubit $\bf B$ is then given by $(\sigma_x+\sigma_z)/\sqrt{2}$. In the ideal case, this yields maximal violation of the classical inequality [Eq.~\eqref{eq:max-violation}]. Panel {\sf\textbf{A.}} of Fig.~\ref{fig:experiment} represents a schematic of the setup and the implemented sequence.

Experimentally, we calibrate the processor for optimal in-sequence re-cooling and gate fidelities, obtaining
\begin{equation}\label{eq:opt-exp-violation}
\Gamma_{\Pr(AB\mid X)}^{\rm exper}=\SI{0.642}{} \pm \SI{0.052}{},
\end{equation}
which is nearly $7\sigma$ below the classical bound of one. Moreover, Eq.~\eqref{eq:fidelity} gives a bound of $F_{\rm Memory}\gtrsim 92\%$ on the quality of the quantum memory in this process.

By adding a \emph{waiting time} after the first measurement and before the unitary, the memory shall lose information over time until it becomes completely classical and, as such, cannot violate classical inequalities. Panel {\bf\textsf{B.}} Fig.~\ref{fig:experiment} shows the degradation of a memory over time towards classicality.

Finally, we address the cross-talk loophole which can occur due to experimental imperfections. First, we look for possible violations of Pearl's inequality [Eqs.~\eqref{eq:Pearl}], which would be a witness of cross-talk influences. The very same data used to obtain Eq.~\eqref{eq:opt-exp-violation} yields
\begin{equation}
\Delta_{\Pr(AB\mid X)}=\SI{0.883}{} \pm \SI{0.023}{},
\end{equation}
showing that Pearl's inequality is not violated, a precondition for no cross-talk influences. Yet, since it does not provide a sufficient condition, we also measured the ACDE [Eq.~\eqref{eq:ACDE}] for which we obtain
\begin{equation}
\mathrm{ACDE}=\SI{0.0325}{} \pm \SI{0.0339}{}.
\end{equation}
This value is compatible with zero within the experimental precision. These address the cross-talk influence loophole\footnote{Importantly, \emph{addressing} the loophole should not be read as \emph{closing} it.}; see the \hyperlink{SM}{Supplemental Information} for additional data.

% \og{In the previous text I made some changes and shortened a bit. You see the changes in the latex, the old text is as a comment there.}

\section{Quantum Properties in Two-Point-Measurement Experiments}

Certification of quantum memory in TPM experiments inevitably hinges on quantum resources elsewhere in the setup. The situation parallels Bell tests: entanglement certification requires incompatible (non–jointly measurable) measurements, while certifying measurements relies on entangled states~\cite{QVB14}. Here, memory certification depends on quantum properties of the initial state $\varrho$ and the unitary $U$. Moreover, these properties are also of independent interest; we therefore outline them below.

In order to certify quantum memory in TPM experiments the initial state $\varrho$ has to be entangled and the measurements incompatible. In the first measurement, it is straightforward to show that, for measure-and-prepare instruments defined by a POVM assemblage $E_{A\mid X}$ and states $\rho_A$, joint measurability of $E_{A\mid X}$ implies classicality. In the second measurement however the notion of incompatibility enters more subtly. Here, the relevant assemblage defines indirect measurements on the system $\bf E$: first, prepare system $\bf A$ in states $\rho_A$, let it interact with $\bf E$ via unitary $U$, and measure $\bf B$ with the POVM $F_B$. This procedure implements the POVM assemblage
\begin{equation}\label{eq:GBA}
G_{B\mid A}=\tr_{\bf A}[U^\dagger(\mathrm{id}_{\bf E^\prime}\otimes F_B)U(\mathrm{id}_{\bf E}\otimes\rho_A)].
\end{equation}
Clearly, $G_{B\mid A}\pos 0$ and $\sum_b G_{B=b\mid A}=\mathrm{id}_{\bf E}$; that is, $G_{B\mid A}$ defines a POVM assemblage on $\bf E$. We show in the \hyperlink{SM}{Supplemental Information} that non-classicality requires the interaction $U$ to \emph{generate} measurement incompatibility; in other words, $G_{B\mid A}$ must not be jointly measurable.

Another property that is certified in the present experiments is \emph{non-classical causality}, here understood as the impossibility of simulating causal processes using classical combinations of common causes and direct causes (see Refs.~\cite{MRSR17,NQGMV21} and references therein). First, as the name suggests, in a common-cause (CC) process, all observed correlations are explainable in terms of a shared common cause. Mathematically, the process is defined by a shared quantum state $\rho_{\bf A^\prime B}$ and not only no information flows from system $\bf B$ to system $\bf A$ (i.e., no retrocausality) but also not in the other direction. In contrast, in DC process all correlations come from direct causal influence. Then, the process is defined by a quantum channel $\mathcal{N}_{\bf A\to B}$. Processes can also be \emph{classical combinations} of these causal mechanisms, i.e., statistical mixtures of CC and DC processes. Every classical process admits such a decomposition, but not all quantum processes do---those that do not are non-classically causal. This property can be certified by violations of the bound in Eq.~\eqref{eq:kmct}. In the panel {\bf\textsf{C}.} of Fig.~\ref{fig:experiment} we plot the observed violations for a process that realises a \emph{coherent mixture of CC and DC}, generated by the two-qubit $\alpha$-\textsc{swap} gate
\begin{equation}
U_{\rm PS}(\alpha)=\cos\left(\frac{\alpha}{2}\right)\mathrm{id}+\ii\sin\left(\frac{\alpha}{2}\right)U_{\textsc{swap}}.
\end{equation}
This generalises the implementation of Ref.~\cite{MRSR17}, recovered at $\alpha=\pi/2$. Here, $\alpha=0$ is a DC process, where $\mathcal{N}_{\bf A\to B}$ is the identity channel; in contrast, $\alpha=\pi$ corresponds to a CC process where $\rho_{\bf A^\prime B}$ is a Bell state. Intermediate values of $\alpha$ implement processes that are not CC, DC, or classical combinations thereof. However, violations appear only for $\pi/2<\alpha<\pi$; outside this range the test is inconclusive since $U_{\rm PS}(\alpha)$ does not generate incompatible assemblages through Eq.~\eqref{eq:GBA} for $\alpha=\pi$ or $0\leq \alpha\leq \pi/2$.

\section{Discussion}

Our work puts forward a DI framework for certifying quantum properties in TPM experiments. The results establish temporal correlations and causal modelling as practical tools for benchmarking state-of-the-art quantum technologies with real-time control via mid-circuit measurements and adaptive operations during computation. The certification occurs under realistic noise and control conditions, does not rely on the devices specific models, details, or calibration, and is versatile in the sense that it can be applied to a variety of existing quantum platforms.

One property that we experimentally certify via temporal correlations is quantum memory. We implement a protocol on a trapped-ion quantum processor that certifies its quantumness and reveals its decay toward classicality. On the one hand, these results show that our methods can be used to establish trust in a quantum memory without requiring a reference memory for benchmarking, since no such reliable device currently exists. On the other hand, the observed agreement between the decay of violations and a simple error model provides evidence that the model captures the relevant noise mechanisms and can therefore serve as a trustworthy witness of classicalisation. Finally, the approach readily extends to other platforms, including photonic interfaces~\cite{Afzelius2015} and recently developed optomechanical memories~\cite{Bozkurt2025}.

This work offers concrete results and opens new directions for exploring dynamical features of quantum information processing from a foundational standpoint. It points to the possibility of benchmarking quantum operations beyond entanglement-breaking and of certifying long-range correlations required to generate entanglement between distant nodes in a network, as well as in algorithms that repeatedly access quantum memories during computation. More broadly, it demonstrates the potential of temporal correlations and causal models as operational tools for connecting foundational insights with the demands of emerging quantum technologies.

\section{Acknowledgments}

We thank Francesco Buscemi, Ties-A. Ohst, Marco T\'ulio Quintino, Roope Uola and Zhen-Peng Xu for discussions.

This research was funded by the European Research Council (ERC, QUDITS, 101039522). Views and opinions expressed are however those of the author(s) only and do not necessarily reflect those of the European Union or the European Research Council Executive Agency. Neither the European Union nor the granting authority can be held responsible for them. We also acknowledge support by the Austrian Science Fund (FWF) through the EU-QUANTERA project TNiSQ (N-6001), by the Austrian Federal Ministry of Education, Science and Research via the Austrian Research Promotion Agency (FFG) through the projects FO999914030 (MUSIQ) and FO999921407 (HDcode) funded by the European Union-NextGenerationEU, by the IQI GmbH, by the Deutsche Forschungsgemeinschaft (DFG, German Research Foundation, project numbers 447948357, 440958198, and 563437167), the Sino-German Center for Research Promotion (Project M-0294), and the German Federal Ministry of Research, Technology and Space (Project QuKuK, Grant No.~16KIS1618K and Project BeRyQC, Grant No.~13N17292). LSVS acknowledges support from the House of Young Talents of the University of Siegen.

\bibliography{bibliography.bib}

\clearpage

\onecolumngrid
\newgeometry{top=1in, bottom=1in, left=1.5in, right=1.5in}
\begin{center}
\vspace*{\baselineskip}
\hypertarget{SM}{\textbf{--- Supplementary Information ---}}\\ \vspace{0.25cm}
{\textbf{\large Device-independent quantum memory certification in two-point measurement experiments}}
\end{center}
\hspace{1cm}

\begin{center}
\begin{minipage}[t]{0.75\textwidth}
In this Supplementary Information, we expand on the discussions presented in the main text by providing detailed proofs of the main results, clarifying experimental procedures and subtleties, and presenting additional illustrative examples. We first introduce classical and quantum models for two-point measurement experiments. We then establish the connection between classical–quantum gaps and the requirement of a genuine quantum memory, showing that violations of the corresponding classical inequalities certify quantum properties and the quality of such unit. Next, we analyse the quantum causal mechanisms underlying non-classicality; in particular we present theoretical description of the $\alpha$-partial \textsc{swap} experiment. We subsequently provide a detailed account of the experimental platform, including the setup, the implemented measurement sequence, and a comprehensive analysis of experimental errors and models. Finally, we present additional data addressing the cross-talk influence loophole.
\end{minipage}
\end{center}
\vspace{0.5cm}

\setcounter{equation}{0}
\renewcommand{\theequation}{S\arabic{equation}}
\renewcommand{\thefigure}{S\arabic{figure}}
\setcounter{page}{1}
\setcounter{figure}{0}
\makeatletter

% \tableofcontents

\setcounter{secnumdepth}{2}

\section*{A. Two-point measurement experiments: From classical to quantum}

\subsection*{Deriving classical inequalities}

In what follows we present a detailed treatment of classical causal models for two-point measurement (TPM) experiments and their statistical implications; this shall enable us to derive some of the results announced in the main text. 

In a TPM experiment we observe two random variables, $A$ and $B$, recording outcomes at successive times. We model them as elements of a classical, time-ordered process in which $A$
may exert a causal influence on $B$, $A\to B$, but not the reverse (thereby excluding the possibility of retrocausal influences). We also allow for latent variables that can affect $A$ and $B$. These unobserved variables are gathered into their own time series, which may interact with the observed outcomes as illustrated in Fig.~\ref{fig:corr-vs-caus}.

The potential outcome model~\cite{Rubin74} is a framework for formalising causal inference by imagining the outcomes that would occur under different interventions. For each unit or system, one defines a set of potential outcomes, one for each possible value of $A$. Causal influence is then understood as the comparison between these outcomes. Formally, a potential outcome model is defined by two ingredients:
\begin{itemize}
    \item[(1)] Two finite-alphabet random variables $A$ and $B$.
    \item[(2)] A collection of random variables $\mathrm{PO}=\{B_{\dop(A=a)}:a\in\mathcal{A}\}$, satisfying
    \begin{equation}\label{eq:POs}
    B=\sum_{a\in\mathcal{A}}\delta_{a}(A)B_{\dop(A=a)}.
    \end{equation}
\end{itemize}
% \og{Is there some normalization missing here? To be discussed.}
Here, $\mathcal{A}$ and $\mathcal{B}$ are the finite alphabets (here, mapped onto sets of real numbers) of $A$ and $B$ respectively and $\delta_y(x)=1$ if $x=y$ and zero otherwise. $A$ and $B$ are referred to as \emph{observational variables}, while the set $\mathrm{PO}$ denotes the collection of \emph{potential outcomes}. Each variable $B_{\dop(A=a)} \in \mathrm{PO}$ represents the outcome that would result if the cause were hypothetically set exogenously to the value $a \in \mathcal{A}$. Accordingly, the direct causation $A\to B$ is captured by how the distribution $\Pr(B_{\dop(A=a)})$ responds to changes in the value of $a$; for instance, the average causal effect (ACE)~\cite{Pearl09},
\begin{equation}
\mathrm{ACE}=\sup|\Pr(B_{\dop(A=a)}=b)-\Pr(B_{\dop(A=a^\prime)}=b)|.
\end{equation}
One can (at least in principle) estimate the joint distribution $\Pr(AB)$ as well as the distributions of individual potential outcomes $\Pr(B_{\dop(A=a)})$, which are referred to as \emph{factual} distributions. Distributions that cannot be estimated in practice (e.g., joint distributions over multiple potential outcomes) are termed \emph{counterfactuals}. 

\begin{figure}
    \centering
    \includegraphics[width=0.75\linewidth]{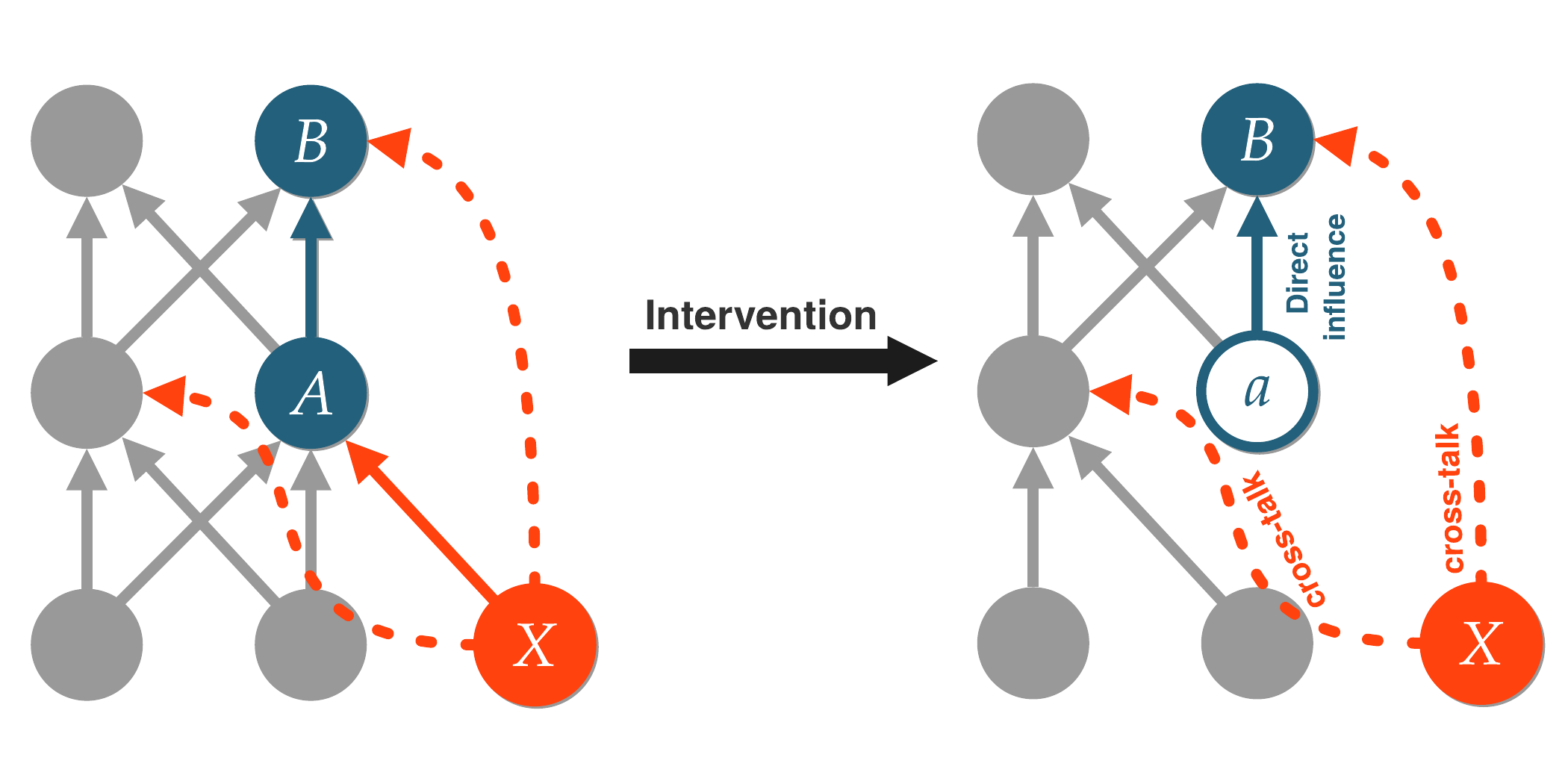}
    \caption{\sf \textbf{Causal influences in TPM experiments.} Bayesian networks representing classical causal models for TPM experiments. The nodes correspond to variables and arrows indicate a potential causal influence. An intervention upon $A$ allows us to access causal influences like direct causal influence $A\to B$ as well as cross-talks.}
    \label{fig:corr-vs-caus}
\end{figure}

Instrumental variables are typically introduced into causal inference motivated by the problem of estimating causal influence without performing interventions~\cite{Pearl09}. For example, if $A$ corresponds to smoking and $B$ to cancer incidence, estimating $\rm ACE$ would involve creating control groups where some people would be forced to smoke, so as to break the influence of potential factors common to smoking and cancer like individual habits or genetics. Another example is when there is not perfect compliance among the individual participants in a treatment trial, where there is a probability that the treatment, for example taking a certain pill, shall be assigned but not actually taken for whatever reason.

Formally, an instrumental variable for a potential outcome model $(A,B,\mathrm{PO})$ is a random variable $X$ with alphabet $\mathcal{X}$ satisfying:
\begin{enumerate}
    \item[(3)] $B_{\dop(A=a,\,X=x)}=B_{\dop(A=a)}$ for all $a\in\mathcal{A}$ and $x\in \mathcal{X}$.
    \item[(4)] $X$ is \emph{jointly independent} on the potential outcomes, $X\perp\!\!\! \perp \mathrm{PO}$.
    % \og{What does "jointly independent" mean? Independent from each of them? Independent to the joint variable?}
\end{enumerate}
Here the symbol ``$\perp\!\!\! \perp$'' denotes statistical independence. Joint independence means statistical independence with joint variables; for instance, saying that $J$ is jointly independent of $K$ and $L$ means $J\perp\!\!\! \perp (K, L)$ and \emph{not} necessarily that $J\perp\!\!\!\perp K$ and $J\perp\!\!\! \perp L$.

\begin{lem}\label{lem:1}
% \og{Is this Lemma known somewhere in the literature?}
If $X$ is an instrumental variable for a potential outcome model $(A,B,\mathrm{PO})$, then
\begin{equation}\label{eq:instrumental-lemma1}
\Pr\big(B_{\dop(A=a)}=b\big) \geq \sup_{x\in\mathcal{X}} \Pr\big(A=a, B=b \mid X=x\big).
\end{equation}
Moreover, if $A$ is binary (i.e., $\mathcal{A}=\{0,1\}$), then
\begin{align}\label{eq:instrumental-lemma2}
\Pr\big(B_{\dop(A=0)}=b_0,\,B_{\dop(A=1)}=b_1\big) \leq \inf_{x\in\mathcal{X}} \Bigl[ &\Pr\big(A=0, B=b_0 \mid X=x\big)\nonumber \\ &+\Pr\big(A=1, B=b_1 \mid X=x\big) \Bigr].
\end{align}
\end{lem}
\begin{proof}
For the first inequality, note that for any $x\in\mathcal{X}$ we have
\begin{align*}
\Pr\big(B_{\dop(A=a)}=b\big)
&\overset{\text{(i)}}{=} \Pr\big(B_{\dop(A=a)}=b \mid X=x\big) \nonumber\\[1mm]
&\overset{\text{(ii)}}{\geq} \Pr\big(A=a, B_{\dop(A=a)}=b \mid X=x\big) \nonumber\\[1mm]
&\overset{\text{(iii)}}{=} \Pr\big(A=a, B=b \mid X=x\big),
\end{align*}
where (i) follows from the independence of $X$ and the potential outcomes, (ii) uses the fact that the marginal probability dominates the corresponding joint probability, and (iii) follows from Eq.~\eqref{eq:POs} together with the assumption that $X$ does not alter the potential outcomes. Taking the supremum over $x\in\mathcal{X}$ proves the first part.

For the second inequality, assume that $A$ is binary. Then, for any fixed $x\in\mathcal{X}$, we can write
\begin{align*}
\Pr\big(B_{\dop(A=0)}=b_0,\,B_{\dop(A=1)}=b_1\big)
&\overset{\text{(v)}}{=} \Pr\big(B_{\dop(A=0)}=b_0,\,B_{\dop(A=1)}=b_1 \mid X=x\big) \nonumber\\[1mm]
&=\Pr\big(A=0, B_{\dop(A=0)}=b_0,\,B_{\dop(A=1)}=b_1 \mid X=x\big) \nonumber\\[1mm]
&\quad + \Pr\big(A=1, B_{\dop(A=0)}=b_0,\,B_{\dop(A=1)}=b_1 \mid X=x\big) \nonumber\\[1mm]
&\overset{\text{(vi)}}{\leq}\Pr\big(A=0, B_{\dop(A=0)}=b_0 \mid X=x\big) \nonumber\\[1mm]
&\quad + \Pr\big(A=1, B_{\dop(A=1)}=b_1 \mid X=x\big) \nonumber\\[1mm]
&=\Pr\big(A=0, B=b_0 \mid X=x\big)+\Pr\big(A=1, B=b_1 \mid X=x\big),
\end{align*}
where (v) is justified by the joint independence $X\perp\!\!\! \perp (B_{\dop(A=0)},B_{\dop(A=1)})$, and the inequality (vi) follows from the properties of probabilities. Taking the infimum over $x\in\mathcal{X}$ completes the proof.
\end{proof}

This lemma allows us to derive a class of instrumental inequalities straightforwardly. First notice that summing up over $b$ in both sides of Eq.~\eqref{eq:instrumental-lemma1} we obtain Pearl's inequality,
\begin{equation}\label{eq:Pearl-ineq-SM}
\max_{a\in\mathcal{A}}\sum_{b\in\mathcal{B}}\sup_{x\in\mathcal{X}}\Pr(A=a,B=b\mid X=x)\leq 1.
\end{equation}
It should be noticed that this inequality is a consequence of the \emph{marginal} independence of the instrumental variable and potential outcomes. It is therefore a \emph{factual} constraint, as both the left and right-hand side of Eq.~\eqref{eq:instrumental-lemma1} are probabilities that can be obtained in experiments (at least in principle). We therefore expect, and it was indeed shown in Ref.~\cite{HLP14}, that no reasonable physical theory violates such a inequality as long the $X$ is an actual instrumental variable. 

To derive inequalities that can be meaningfully tested and potentially violated by quantum statistics, it is essential to incorporate counterfactual constraints. Notably, the assumption of joint independence constitutes one such constraint. In particular, by summing over $b_0$ and $b_1$ on both sides of inequality~\eqref{eq:instrumental-lemma2}, we obtain
\begin{equation}\label{eq:KM-SM}
\sum_{b_0,\,b_1\in\mathcal{B}}\inf_{x\in \mathcal{X}}\sum_{a\in\{0,1\}}\Pr(A=a, B=b_a\mid X=x)\geq 1.
\end{equation}
This is the inequality we use to certify quantum properties. Other instrumental inequalities can be found, for example, in the recent work of K\'edagni and Mourifi\'e~\cite{KM20} and the pioneering work of Chaves \emph{et al.}~\cite{Chaves2017}; see also Refs.~\cite{VanHimbeeck2019,Gachechiladze2020,Agresti2022}.

\subsection*{Deriving a quantum inequality}

Now we shall discuss quantum violations of inequality~\eqref{eq:KM-SM}. For this purpose, mapping between the instrumental network (i.e., the network of Fig.~\ref{fig:corr-vs-caus} in the absence of crosstalk) into local hidden-variable (LHV) models for Bell tests shall be very useful, since many results are already established for the latter.

To obtain this mapping we first group the latent variables that jointly influence $A$ and $B$ into a single variable $\Lambda$ so that probabilities factorise as in
\begin{equation}
\Pr(AB\mid X)=\sum_{\lambda} \Pr(\Lambda=\lambda)\Pr(A\mid X,\Lambda=\lambda)\Pr(B\mid  A, \Lambda=\lambda).
\end{equation}
From this, it is obvious that such probabilities can be formally written as post-selection probabilities of a LHV model for a Bell test,
\begin{equation}\label{eq:postselect-bell}
\Pr(AB\mid X)=\delta(Y,A)\Pr(AB\mid X,Y=A)_{\rm Bell},
\end{equation}
where Bell's probabilities in a LHV model it takes the form
\begin{equation}
\Pr(AB\mid XY)_{\rm Bell}=\sum_{\lambda}\Pr(\Lambda=\lambda)\Pr(A\mid X,\Lambda=\lambda)\Pr(B\mid  Y, \Lambda=\lambda).
\end{equation}
Indeed, this interrelation lies behind violations of instrumental inequalities reported in Ref.~\cite{Chaves2017}. Yet. as mentioned in the main text and also discussed below, although we shall use this map as a tool, this is not how we envision nor implement such experiments in practice.

One can define the variables $\hat{A}_X$ (and $\hat{B}_Y$) by the rule that they assume value $+1$ if $A_{\dop(X)} = 0$ ($B_{\dop (Y)}=0$, respectively) and $-1$ otherwise. Then, one can write Bell's probabilities as
\begin{equation}\label{eq:correlations-bell}
\Pr(AB\mid XY)_{\rm Bell}=\frac{1}{4}\left[1+(-1)^A\langle \hat{A}_X\rangle+(-1)^B\langle \hat{B}_Y\rangle+(-1)^{A+B}\langle \hat{A}_X\hat{B}_Y\rangle \right].
\end{equation}
Moreover, one can define the CHSH scores
\begin{equation}\label{eq:CHSH}
\mathrm{CHSH}=\langle \hat{A}_{x} \hat{B}_{y}\rangle+\langle \hat{A}_{x^\prime} \hat{B}_{y}\rangle+\langle \hat{A}_{x} \hat{B}_{y^\prime}\rangle-\langle \hat{A}_{x^\prime} \hat{B}_{y^\prime}\rangle,
\end{equation}
and equivalent ones using symmetry transformations. Then we have:

\begin{lem}\label{lem:2}
Let $X$ be an instrumental variable for a potential outcome model $(A,B,\mathrm{PO})$, with $A$ and $B$ binary. Then
\begin{equation}\label{eq:IKM-CHSH}
\sum_{b_0,\,b_1\in\mathcal{B}}\inf_{x\in \mathcal{X}}\sum_{a\in\{0,1\}}\Pr(A=a, B=b_a\mid X=x)
=2-\frac{\mathrm{CHSH}^{\prime}+\mathrm{CHSH}^{\prime\prime}}{4},
\end{equation}
where the right-hand side involves the sum of two CHSH scores defined in Eq.~\eqref{eq:CHSH}.
\end{lem}
\begin{proof}
Define
\begin{align}
\mathcal{I}
&=\sum_{b_0,\,b_1\in\mathcal{B}} \inf_{x\in \mathcal{X}} \sum_{a\in\{0,1\}} \Pr(A=a, B=b_a\mid X=x)\nonumber\\
&=\sum_{b_0,\,b_1\in\mathcal{B}}\Big[\Pr(A=0,B=b_0\mid X=f(b_0,b_1))+\Pr(A=1,B=b_1\mid X=f(b_0,b_1))\Big],
\end{align}
for some function $f:\mathcal{B}\times\mathcal{B}\to\mathcal{X}$.  
Let $x_{b_0b_1}=f(b_0,b_1)$. Since $B$ is binary, we can expand
\begin{align}
\mathcal{I}
&=\Pr(A=0,B=0\mid X=x_{00})+\Pr(A=1,B=0\mid X=x_{00}) \nonumber\\
&\quad+\Pr(A=0,B=1\mid X=x_{11})+\Pr(A=1,B=1\mid X=x_{11}) \nonumber\\
&\quad+\Pr(A=0,B=0\mid X=x_{01})+\Pr(A=1,B=1\mid X=x_{01}) \nonumber\\
&\quad+\Pr(A=0,B=1\mid X=x_{10})+\Pr(A=1,B=0\mid X=x_{10}).
\end{align}
Substituting Eqs.~\eqref{eq:postselect-bell} and \eqref{eq:correlations-bell} into this expression, all single-party terms cancel, giving
\begin{align}\label{eq:Lemma-CHSH}
\mathcal{I}
&=1+\frac{-\langle \hat{A}_{x_{00}}\hat{B}_0\rangle+\langle \hat{A}_{x_{00}}\hat{B}_1\rangle-\langle \hat{A}_{x_{11}}\hat{B}_0\rangle-\langle \hat{A}_{x_{11}}\hat{B}_1\rangle}{4}\nonumber\\
&\quad+1+\frac{\langle \hat{A}_{x_{01}}\hat{B}_0\rangle-\langle \hat{A}_{x_{01}}\hat{B}_1\rangle+\langle \hat{A}_{x_{10}}\hat{B}_0\rangle+\langle \hat{A}_{x_{10}}\hat{B}_1\rangle}{4},
\end{align}
which is precisely Eq.~\eqref{eq:IKM-CHSH}.
\end{proof}

The quantum bound (Tsirelson's bound) for the CHSH inequality is well-known, being $|\mathrm{CHSH}|\leq 2\sqrt{2}$. From this, we conclude that quantum violations of the instrumental inequality, if they exist, are limited to
\begin{equation}
\sum_{b_0,\,b_1\in\mathcal{B}}\inf_{x\in \mathcal{X}}\sum_{a\in\{0,1\}}\Pr(A=a, B=b_a\mid X=x)
\overset{\rm (Q)}{\leq} 2-\sqrt{2}\approx 0.58579.
\end{equation}
We shall show latter when discussing quantum models for TPM experiments that this bound is indeed tight.

\subsection*{Addressing cross-talk influences}

Another consequence of Lemma~\ref{lem:2} is that it enables a natural extension of classical models to situations with cross-talk influences, by using probabilistic couplings in the spirit of quantum contextuality with disturbing measurements~\cite{KDL15}, which is often used to rigorously analyse non-classicality in the presence of imperfections in the form of signalling or measurement disturbance. Here, we use this result by identifying cross-talk in TPM experiments as one-way signalling in a Bell test.

The idea is best motivated by considering the bipartite Bell test. Recall first that Bell locality can be formulated as a classical marginal problem: given some pairwise joint distributions, the question is whether there exists a global distribution that reproduces them as marginals. Formally~\cite{Fine1982}, Bell locality holds if and only if one can find a set of jointly distributed random variables $\tilde{A}_0\tilde{A}_1\dots \tilde{B}_1\tilde{B}_2\dots$ (a \textit{coupling}) such that
\begin{equation}
\Pr(\tilde{A}_X\tilde{B}_Y)=\Pr(AB\mid XY)
\end{equation}
The key insight of Ref.~\cite{KDL15} is to extend this construction to cases where signalling is present. In such situations we must regard the outcomes as context-dependent, i.e. $A_{X,Y}$ and $B_{X,Y}$; that is, the outcome of one part is causally influenced by the choice of setting in the other. One then considers couplings of jointly distributed variables $\tilde{A}_{X,Y}$ and $\tilde{B}_{X,Y}$ chosen so that the probabilities 
\begin{equation*}
\Pr(\tilde{A}_{X,Y=y}=\tilde{A}_{X,Y=y^\prime})\quad\quad\text{and}\quad\quad\Pr(\tilde{B}_{X=x,Y=y}=\tilde{B}_{X=y^\prime,Y}),
\end{equation*}
are as large as possible given the observed marginals,
\begin{equation}
\Pr(\tilde{A}_{X,Y}=A_{X,Y})=\Pr(\tilde{B}_{X,Y}=B_{X,Y})=1.
\end{equation}
In other words, the ``copies'' of the same measurement across different contexts are made maximally consistent with one another. If such a maximal coupling exists, then any observed incompatibility with LHV models cannot be explained away by signalling alone.

In the scenario under analysis, we need only consider signalling from one party to the other, since retro-causality has been excluded \emph{a priori}. By adapting one of the main results of Ref.~\cite{KDL15}, the left-hand side of Eq.~\eqref{eq:Lemma-CHSH} can be expressed as
\begin{align*}
\sum_{b_0,,b_1\in\mathcal{B}}\inf_{x\in \mathcal{X}}\sum_{a\in{0,1}}\Pr(A=a, B=b_a\mid X=x)
\geq 2-\tfrac{1}{4}\Big(&\mathrm{CHSH}^{\prime}+\mathrm{CHSH}^{\prime\prime} \\
&+4\sup\big|\langle B_{\dop(A=a,X=x)}\rangle-\langle B_{\dop(A=a,X=x^\prime)}\rangle\big|\Big),
\end{align*}
assuming $A$ and $B$ binary. Equivalently, this can be written as
\begin{equation}
\sum_{b_0,\,b_1\in\mathcal{B}}\inf_{x\in \mathcal{X}}\sum_{a\in\{0,1\}}\Pr(A=a, B=b_a\mid X=x)+2\mathrm{ACDE}\geq 1,
\end{equation}
leading to Eq.~(3) in the main text. Here, we have introduced the average causal direct effect (ACDE),
\begin{equation}
\mathrm{ACDE}=\sup|\Pr(B_{\dop(A=a,X=x)}=b)-\Pr(B_{\dop(A=a,X=x^\prime)}=b)|.
\end{equation}
The ACDE quantifies cross-talk influences, i.e. direct causal effects $X\to B$ that are not mediated by $A$. Importantly, this quantity is empirically accessible and provides a way to separate genuine non-classical phenomena from spurious violations attributable to cross-talks. 

% \og{I understand the main idea of the part starting from Eq. (A7) as follows: We want to discuss influence and noise. For that, we first derive the inequality (A11), which does not need big assumptions, but it contains CHSH scores, which are well known from Bell inequalities. 
% For CHSH scenarios, influences / cross talk / signalling has been discussed in Ref. [XXX]. So we can use these approaches, with the additional simplification that we only have to consider signalling
% in one direction. This leads to Eq. A17, which is finally testable.

% Is my understanding correct? Then, we should maybe start a subsection
% before A7 and start with such an overview, to explain the goal of the subsection.}

\subsection*{Quantum two-point measurement experiments}

Finally, we now explain how TPM experiments can be described within quantum theory, in particular how they can lead to violations of the classical bound of one for inequality~\eqref{eq:KM-SM}.

Recall that a TPM experiment starts with a system $\bf A^\prime$ that is potentially correlated with a system $\bf E$. First, system $\bf A^\prime$ is locally measured with a quantum instrument, which involves a general measurement where not only the classical outcome is considered but also the post-measurement state. We shall work in the Choi-Jamio{\l}kowski (CJ) picture~\cite{Choi75,J72}, wherein an instrument is represented by a set of positive semidefinite operators $\mathcal{E}_{A}\pos 0$ that satisfy $\tr_{\bf A}\sum_{a}\mathcal{E}_{A=a}=\mathrm{id}_{\bf A^\prime}$. The event $A=a$ occurs with probability $\tr \sigma_{A=a}$, leaving the joint system $\bf A+E$ in the state $(\tr \sigma_{A=a})^{-1}\sigma_{A=a}$,
\begin{equation}
\sigma_A\equiv \tr_{\bf A^\prime}[(\mathrm{id}_{\bf A}\otimes \varrho^{\Gamma_{\bf A^\prime}})(\mathcal{E}_{A}\otimes \mathrm{id}_{\bf E})].
\end{equation}
The symbols $\tr_{\bf X}$ and $\Gamma_{\bf X}$ denote partial trace and partial transposition over a system $\bf X$, respectively, and $\mathrm{id}_{\bf X}$ is the identity operator on the Hilbert space of $\bf X$, $\mathcal{H}_{\bf X}$. The second measurement can be regarded as a positive-operator-value measure (POVM) on $\bf B$, $F_B\pos 0$ and $\sum_b F_{B=b}=\mathrm{id}_{\bf B}$ (this is because no further measurements are of interest, so we may assume, without loss of generality, the system is discarded after the outcome $B$ is observed).

The joint statistics of $A$ and $B$ is computed via a Born's-like rule,
\begin{equation}
\Pr(AB\mid X)=\tr[(\mathcal{E}_{A\mid X}\otimes F_B) W]
\end{equation}
Here, $W$ is the {\emph{process operator}}, defined in terms of $\varrho$ and $U:\bf \mathcal{H}_{A}\otimes \mathcal{H}_E\to \mathcal{H}_B\otimes \mathcal{H}_{E^\prime}$ as
\begin{equation}
W=\tr_{\bf EE^\prime}[(\varrho^{\Gamma_{{\bf E}}}\otimes \mathrm{id}_{\bf ABE^\prime})(\mathrm{id}_{\bf A^\prime}\otimes |U\rangle\!\rangle\langle\!\langle U|)],
\end{equation}
where $|U\rangle\!\rangle$ denotes the vectorisation of $U$. In other words, the process operator completely determines the statistics of TPM experiments, analogously to a density operator for prepare-and-measure or Bell experiments~\cite{WMGWF15}.

To avoid cross-talk, it is essential that the quantum instruments $\mathcal{E}_{A\mid X}$ be such that the result $A$ \emph{separates} the variables $X$ and $B$. This means that the action of setting $X$ to different values cannot leak through the environment or be hidden in the post-measurement state of $\bf A$. In fact, this is possible only if $\mathcal{E}_{A\mid X}$ is a measure-and-prepare instrument of the form
\begin{equation}
\mathcal{E}_{A\mid X}=E_{A\mid X}\otimes \rho_A,
\end{equation}
where $E_{A\mid X}$ is a POVM assemblage (i.e., $E_{A\mid X}\pos 0$ and $\sum_{a} E_{A=a\mid X}=\mathrm{id}_{\bf A^\prime}$) acting upon $\bf A^\prime$ and $\rho_{A=a}$ is the normalised post-measurement state of system $\bf A$ conditioned upon the event $A=a$. Importantly, the states are allowed to depend only on the observed outcome $A=a$ and \emph{not} on the valued setted to $X$.

Within this restrictions, one can define the observational probabilities,
\begin{equation}\label{eq:QuantumTPMInstrumental}
\Pr(AB\mid X)=\tr[(E_{A\mid X}\otimes \rho_A\otimes F_B) W],
\end{equation}
as well as the intervention $\dop(A=a)$, which yields
\begin{equation}\label{eq:Quantumdop}
\Pr(B_{\dop(A=a)}=b)=\tr[(\mathrm{id}_{\bf A^\prime}\otimes \rho_{A=a}\otimes F_{B=b}) W]
\end{equation}
To gain some insight, recall that the identity matrix in the CJ isomorphism corresponds to the partial trace. Therefore, intervention on the process essentially acts on the instrument $\mathcal{E}_{A\mid X}$, discarding (i.e., tracing out) the input and preparing the state $\rho_A$ of interest in an exogenous manner (i.e., by means other than measuring $\bf A^\prime$).

This quantum TPM model is consistent with the classical case, taking the latter as a particular case. Moreover, notice that $\mathrm{id}_{\bf A^\prime}\pos E_{A\mid X}$ implies Eq.~\eqref{eq:instrumental-lemma1} of Lemma~\ref{lem:1},
\begin{equation}
\sup_{x\in\mathcal{X}}\Pr(A=a,B=b\mid X=x)\geq \Pr(B_{\dop(A=a)}=b)\quad \forall x\in\mathcal{X},
\end{equation}
which is directly leads to Pearl's inequality~\eqref{eq:Pearl-ineq-SM}. The other instrumental inequality~\eqref{eq:KM-SM}, in turn, can be violated by quantum TPM correlations. In fact, the process operator
\begin{equation}\label{eq:W222}
W_{222}=|\mathrm{GHZ}\rangle\langle\mathrm{GHZ}|+\sigma_x^{\bf A}|\mathrm{GHZ}\rangle\langle\mathrm{GHZ}|\sigma_x^{\bf A},
\end{equation}
introduced in Ref.~\cite{NQGMV21} to as the ``maximally robust'', maximally violates inequality~\eqref{eq:KM-SM}. Here, $|\mathrm{GHZ}\rangle\propto |000\rangle+|111\rangle$ is the GHZ state and $\sigma_x^{\bf A}$ is the X-Pauli matrix acting on $\bf A$. We shall discuss this violation in more details in the next sections as this was one the models we implement experimentally.

\section*{B. Quantum memory certification}

The purpose of this section is to examine how TPM correlations can be used to certify the correct operation of a quantum memory.

First, recall that a quantum memory stores a quantum state and later returns it on demand, preserving coherence and entanglement with external systems. Ideally, the overall write–store–read operation acts as the identity channel on the stored system (or as a fixed unitary that can be undone), making the memory perfectly reversible. In practice, imperfections are modelled by a noisy quantum channel; i.e., a completely positive, trace-preserving (CPTP) map capturing irreversible effects such as losses, decoherence and dissipation.

The degradation of a quantum channel can lead it towards classicality. When this occurs, the channel $\mathcal{M}$ becomes entanglement-breaking~\cite{HSR03}---i.e., it destroys any entanglement the input system may have had with other degrees of freedom. Equivalently, the channel becomes a measure-and-prepare channel,
\begin{equation}
\mathcal{M}^{\rm EB}(\,\cdot\,)=\sum_\lambda \tr(E_\lambda \,\cdot\,)\, \rho_\lambda
\end{equation}
Here, $E_\lambda\pos 0$ and $\sum_\lambda E_\lambda=\mathrm{id}$ is a POVM and $\rho_\lambda$ are quantum states. When that happens, the device can only store classical information contained within $\lambda$ about the input state. Although entanglement-breaking channels may not be considered classical channels from certain perspectives~\cite{Vieira2025}, they can certainly be regarded as such when discussing memories, since the information storage mechanism is entirely classical in this case.

\begin{prop}\label{prop:QM}
If a TPM process has no quantum memory, then
\begin{equation}\label{eq:WCM}
W=\sum_{\lambda}\mu_\lambda \, \eta_{\bf A^\prime}^\lambda \otimes N_{\bf A\to B}^\lambda,
\end{equation}
where $\mu_\lambda\geq 0$ with $\sum_{\lambda}\mu_\lambda=1$, each $\eta_{\bf A^\prime}^\lambda$ is a density operator, and $N_{\bf A\to B}^\lambda$ is the Choi--Jamiołkowski representation of a channel from $\bf A$ to $\bf B$. Moreover, quantum TPM correlations of the form of Eq.~\eqref{eq:QuantumTPMInstrumental} and \eqref{eq:Quantumdop} are compatible with a classical causal model.
\end{prop}
\begin{proof}
If a TPM process has no quantum memory, then the initial state $\varrho_{\bf A^\prime E}$ is mapped as $\varrho_{\bf A^\prime E}\;\mapsto\;\mathcal{M}^{\rm EB}_{\bf E\to E}(\varrho_{\bf A^\prime E})$, where $\mathcal{M}^{\rm EB}_{\bf E\to E}$ is an entanglement-breaking channel. Hence,
\begin{equation}
\mathcal{M}^{\rm EB}_{\bf E\to E}(\varrho_{\bf A^\prime E})
=\sum_{\lambda}\mu_\lambda\, \varrho_{\bf A^\prime}^\lambda\otimes \varrho_{\bf E}^\lambda.
\end{equation}
The TPM process operator is then
\begin{align}
W_{\rm CM}
&=\tr_{\bf EE^\prime}\Big[\big(\mathcal{M}^{\rm EB}_{\bf E\to E}(\varrho_{\bf A^\prime E})\otimes \mathrm{id}_{\bf ABE^\prime}\big)|U\rangle\!\rangle\langle\!\langle U|\Big]\nonumber \\
&\overset{(*)}{=}\tr_{\bf E}\Big[\big(\mathcal{M}^{\rm EB}_{\bf E\to E}(\varrho_{\bf A^\prime E})\otimes \mathrm{id}_{\bf AB}\big)K_{\bf AE\to B}\Big]\nonumber \\
&=\sum_{\lambda}\mu_\lambda\, \varrho_{\bf A^\prime}^\lambda\otimes \tr_{\bf E}\big[(\varrho_{\bf E}^\lambda\otimes \mathrm{id}_{\bf AB})K_{\bf AE\to B}\big],
\end{align}
where in $(*)$ we have defined $K_{\bf AE\to B}=\tr_{\bf E^\prime}\,|U\rangle\!\rangle\langle\!\langle U|$. Next, define
\begin{equation}
N_{\bf A\to B}^\lambda=\tr_{\bf E}\big[(\varrho_{\bf E}^\lambda\otimes \mathrm{id}_{\bf AB})K_{\bf AE\to B}\big].
\end{equation}
For each $\lambda$, we have $N_{\bf A\to B}^\lambda\succcurlyeq 0$ and
\begin{align}
\tr_{\bf B}\,N_{\bf A\to B}^\lambda
&=\tr_{\bf E}\Big[(\varrho_{\bf E}^\lambda\otimes \mathrm{id}_{\bf AB})\tr_{\bf B}\, K_{\bf AE\to B}\Big]\nonumber \\
&=\mathrm{id}_{\bf A}\,\tr \varrho_{\bf E}^\lambda\nonumber \\
&=\mathrm{id}_{\bf A},
\end{align}
since $\tr_{\bf B}\, K_{\bf AE\to B}=\mathrm{id}_{\bf AE}$. Thus, each $N^\lambda_{\bf A\to B}$ is indeed the CJ operator of a channel $\bf A\to B$. 

Finally, quantum TPM probabilities take the form
\begin{align}
\Pr(AB\mid X)
&=\tr\!\left[(E_{A\mid X}\otimes \rho_A\otimes F_B)\,W\right]\nonumber \\
&=\sum_\lambda \mu_\lambda\,\tr\!\big(E_{A\mid X}\,\eta^\lambda_{\bf A^\prime}\big)\,\tr\!\big[(\rho_A\otimes F_B)\,N^\lambda_{\bf A\to B}\big].
\end{align}
This naturally suggests the definitions $\Pr(\Lambda=\lambda)=\mu_\lambda$, $\Pr(A\mid X,\Lambda=\lambda)=\tr(E_{A\mid X}\,\eta^\lambda_{\bf A^\prime})$,
$\Pr(B\mid A,\Lambda=\lambda)=\tr[(\rho_A\otimes F_B)\,N^\lambda_{\bf A\to B}]$. These are valid probability distributions. They are obviously positive and the normalisation of $\Pr(B\mid A,\Lambda=\lambda)$ follows from the fact that each $N^\lambda_{\bf A\to B}$ is the CJ operator of a channel:
\begin{align}
\sum_{b}\Pr(B=b\mid A,\Lambda=\lambda)
&=\tr\big[(\rho_A\otimes \mathrm{id}_{\bf B})\,N^\lambda_{\bf A\to B}\big]\nonumber \\
&=\tr\!\left(\rho_A\,\tr_{\bf B}N^\lambda_{\bf A\to B}\right)\nonumber \\ &=1,
\end{align}
which works only because $\tr_{\bf B}\, N^\lambda_{\bf A\to B}=\mathrm{id}_{\bf A}$. Quantum do-probabilities in turn read
\begin{align}
\Pr(B_{\dop(A=a)}=b)&=\tr[(\mathrm{id}_{\bf A^\prime}\otimes \rho_{A=a}\otimes F_{B=b})W]\nonumber\\ &=\sum_\lambda \mu_\lambda\tr[(\rho_{A=a}\otimes F_{B=b})N_{\bf A\to B}^\lambda] \nonumber\\ &=\sum_\lambda \Pr(\Lambda=\lambda)\Pr(B=b\mid A=a,\Lambda=\lambda).
\end{align}
This concludes the proof.
\end{proof}

Therefore, the violation of a classical causal inequality proves that a quantum memory cannot be simulated as a classical classical memory acting as a entanglement-breaking channel. 

One way to quantify the quality of a memory is in terms of its fidelity to an ideal memory. To do this, consider $M$ the normalised Choi-Jamio{\l}kowski (CJ) operator of a quantum channel,
\begin{equation}
M=(\mathrm{id}\otimes \mathcal{M})(|\mathrm{id}\rangle\langle\mathrm{id}|),
\end{equation}
where
\begin{equation}
\ket{\mathrm{id}}=\frac{1}{\sqrt{d}}\sum_i \ket{i}\otimes \ket{i},
\end{equation}
and $d$ denotes the system dimension. We define the fidelity to an ideal quantum memory (here, the identity channel with no loss of generality) by
\begin{equation}
F_{\rm Memory}=\bra{\mathrm{id}}M\ket{\mathrm{id}}.
\end{equation}
For the case of qubit memories, this quantity can be estimated by invoking a self-testing bound proved by Jędrzej Kaniewski~\cite{Kaniewski16}. Such a bound allows us to estimate fidelity of a two-qubit state to the ideal Bell state can be DI lower bounded in a Bell test by CHSH violations. This result together with the mapping of Eq.~\eqref{eq:IKM-CHSH} and the CJ isomorphism lead us to the following bound
\begin{equation}
F_{\rm Memory}\geq\frac{1}{2}\left[1-\frac{1}{\sqrt{2}-S_{\rm K}}\left(\sum_{b_0,b_1}\inf_{x\in \mathcal{X}}\sum_{a}\Pr(A=a,B=B_0\mid X=x)-2+S_{\rm K}\right)\right],
\end{equation}
where $S_{\rm K}=(8+7\sqrt{2})/17$.

Finally, from the structure of the process operator, we see that a classical memory (i.e., an entanglement-breaking channel) yields a specific form for $W$, as characterised by a constrained separability problem [Eq.~\eqref{eq:WCM}]. In particular, note that in Eq.~\eqref{eq:WCM} the operators $N_{\bf A\to B}^\lambda$ must not only be positive, as in the standard separability problem, but must also satisfy the marginal constraints $\tr_{\bf B}, N_{\bf A\to B}^\lambda=\mathrm{id}_{\bf A}$ for all $\lambda$. This imposes substantial differences from ordinary entanglement structure in the process operator $W$; see Refs.~\cite{Giarmatzi2021,NQGMV21,santos2025two} for an in-depth discussion. Notably, Ref.~\cite{NQGMV21} found an instance of a TPM process operator that is separable in the usual sense yet cannot be decomposed in the form of Eq.~\eqref{eq:WCM}.

It is thus natural to ask whether our results allow for a sharper distinction—namely, whether there exists a gap between mere entanglement in a TPM process operator and the genuinely quantum memory required to violate classical inequalities. Although we were unable to find a quantum violation for the process introduced in Ref.~\cite{NQGMV21}, we identified another process with the same separability properties which, crucially, \emph{maximally violates the classical bound}. Specifically, define
\begin{equation}
\Upsilon_p = p\,W_{222} + (1-p)\,V_{\bf A}\,W_{222}\,V_{\bf A}^\dagger,
\end{equation}
where $0 \leq p \leq 1$, $V_{\bf A}=H\sigma_x$ is the composition of the X-Pauli and Hadamard gates on qubit $\bf A$, and $W_{222}$ is the process introduced in Eq.~\eqref{eq:W222}. A straightforward calculation shows that this family yields maximal quantum violations for all $p$ in the interval $0\leq p\leq 1$, yet for $p=1/2$ the process is separable. To the best of our knowledge, this provides the strongest demonstration to date that the marginal constraints in the decomposition of Eq.~\eqref{eq:WCM} have strong and physically meaningful consequences for TPM processes.

\section*{C. Quantum causal mechanisms}

In the quantum TPM scenario, we saw a gap between classical and quantum predictions. It is therefore natural to inquire about quantum mechanisms behind such mismatch. Here, we explore properties in the causal mechanisms of quantum models that differentiate them from classical counterparts.

Correlations between two variables, here $A$ and $B$, are typically explained in terms of common-causes (CC) and direct causes (DC). In a pure quantum CC scenario, all correlations between $A$ and $B$ are explainable in terms of shared non-signalling resources; that is, a shared state $\varrho$. If this state is entangled, we have a non-classical CC, which can lead, for example, to the violation of a Bell inequality. 

Entanglement as a synonym of quantum CC is a well established notion in literature; see, e.g., Ref.~\cite{Wolfe2020} for in-depth discussion. In contrast, for DC there is no definition with a similar status. Here, we identify non-classicality in DC as the possibility to generate incompatibility in Bob's measurements; that is, when there exists states $\rho_A$ for the system $\bf \bar{A}$ and some POVM $F_B$ on the system $\bf B$ such that the \emph{assemblage} of POVMs defined via
\begin{equation}\label{eq:GBASM}
G_{B\mid A}=\tr_{\rm A}[U^\dagger(\mathrm{id}_{\bf E^\prime}\otimes F_B)U(\mathrm{id}_{\bf E}\otimes\rho_A)]
\end{equation}
is incompatible (i.e., not jointly measurable). In other words, one \emph{cannot} observe a classical-quantum gap in TPM experiments if for all $\rho_A$ and $F_B$, $G_{B\mid A}$ can be implemented with a single POVM followed by classical post-processing; that is, if one can find a POVM $H_\lambda$ and a classical channels $h(B\mid X \Lambda)$ such that
\begin{equation}\label{eq:quantumDC}
G_{B\mid A}=\sum_{\lambda} h(B\mid A, \Lambda=\lambda)H_\lambda.
\end{equation}
Measurement incompatibility is also required in the POVM $E_{A\mid X}$ [Eq.~\eqref{eq:QuantumTPMInstrumental}].

\begin{figure}
    \centering
    \includegraphics[width=0.9\linewidth]{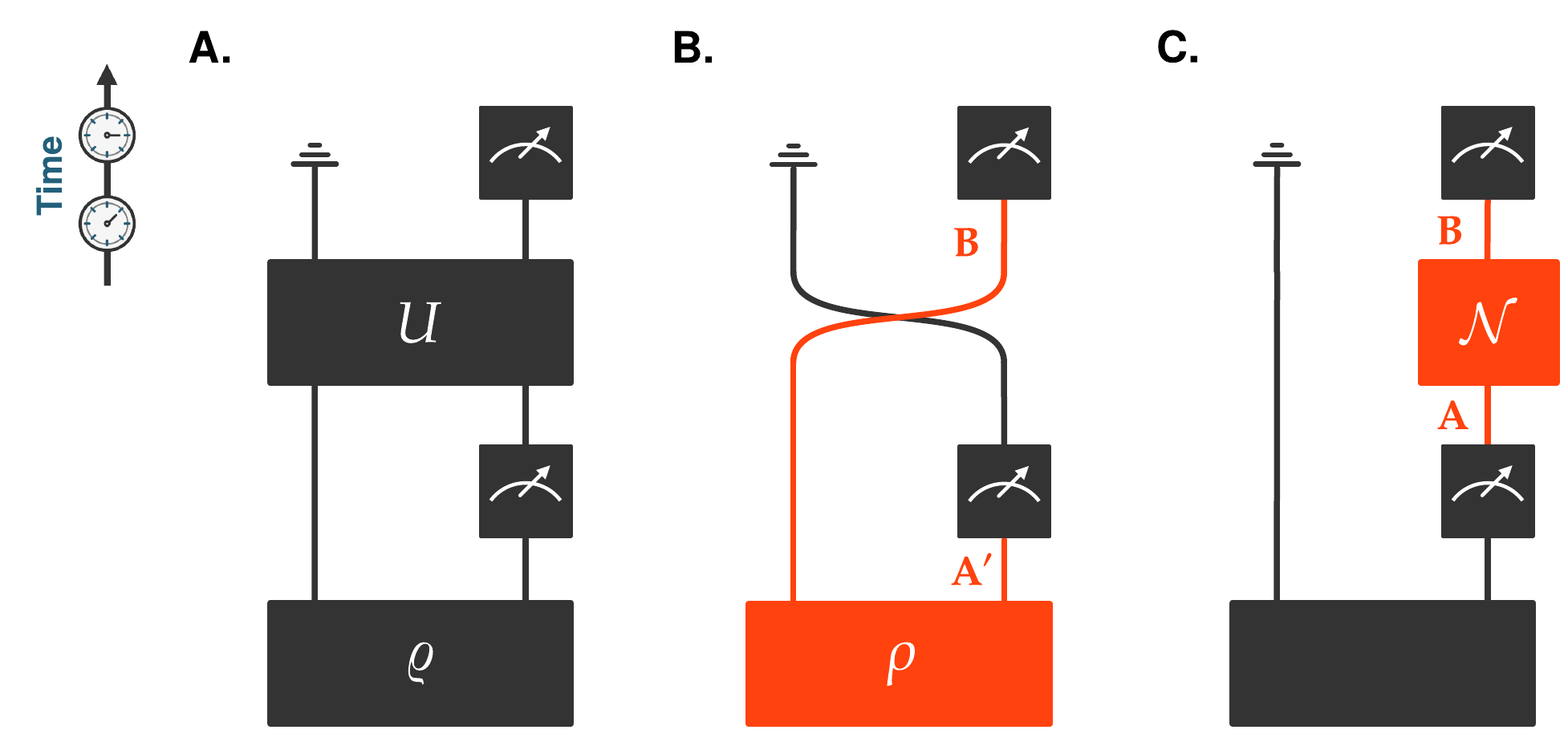}
    \caption{\sf \textbf{Quantum common-causes and direct-causes.} {\bf\textsf{A.}} A general TPM process synthesises both common-cause and direct-cause mechanisms, encapsulated either by a pair $(\varrho, U)$ or, from an agent perspective, by a process operator $W$. {\bf\textsf{B.}} Common-cause (CC) processes are those in which all observed correlations originate from a shared quantum state $\rho_{\bf A^\prime B}$ (should not be confused with $\varrho$ in panel {\bf\textsf{A.}}). {\bf\textsf{C.}} Direct-cause (DC) processes, by contrast, are those in which all observed correlations arise from the direct influence $\bf A \to B$, mediated by a quantum channel $\mathcal{N}_{\bf A\to B}$.}
    \label{fig:CC-DC}
\end{figure}

From the agents' perspective, all available information that can be obtained is the process operator (see Fig.~\ref{fig:CC-DC}). In this context, they identify the process as CC-type when there is no detectable direct causal influence; that is, the process matrix only encodes shared correlations,
\begin{equation}\label{eq:classical_CCDC}
W_{\rm CC}=\rho_{\bf A^\prime B} \otimes \mathrm{id}_{\bf A},
\end{equation}
where $\rho_{\bf A^\prime B}$ is a state the combined system $\bf A^\prime+B$. In those cases, all statistical correlations are explained by the sharing of correlations encoded in the bipartite state $\rho_{\bf A^\prime B}$, without invoking any DC influence. In contrast, the process is DC-type when statistical correlations result solely from direct causation, without any influence from confounding factors. The process operator for a DC-type process takes the form
\begin{equation}
W_{\rm DC}=\rho_{\bf A^\prime}\otimes N_{\bf A\to B},
\end{equation}
where $N_{\bf A\to B}$ represents a channel (i.e., the CJ operator of a CPTP map $\mathcal{N}_{\bf A\to B}$) connecting system $\bf A$ to system $\bf B$, and $\rho_{\bf A^\prime}$ a state for system $\bf A^\prime$. If a process is such that the process operator belongs to the convex hull of CC-type and DC-type processes, i.e.,
\begin{equation}
W_{\rm CCDC}=pW_{\rm CC}+(1-p)\sum_{\lambda}P(\lambda)W_{\rm DC}^{\lambda},
\end{equation}
$P(\lambda)\geq 0$, $\sum_{\lambda}P(\lambda)=1$, and $0\leq p\leq 1$, then we say that the CC-DC combination is classical. Note that individually, both CC and DC can be non-classical from the perspective of those who know $\varrho$ and $U$, but their combination from the agents' perspective can be classical, i.e., possible to simulate with the help of a classical mechanism.

In summary, we identify three distinct notions of quantumness. Namely, for a given pair $(\varrho,U)$ implementing a TPM process, we say that
\begin{enumerate}
    \hypertarget{C1}{\item[(C1)]} the CC is classical if $\varrho$ is separable (i.e., non-entangled);
    \hypertarget{C2}{\item[(C2)]} the DC is classical if $U$ does \emph{not} induce incompatibility via Eq.~\eqref{eq:quantumDC};
    \hypertarget{C3}{\item[(C3)]} the relation between CC and DC is classical if the process operator $W$ can be decomposed as in Eq.~\eqref{eq:classical_CCDC}.
\end{enumerate}

\begin{prop}
Any TPM process, characterised by a process operator $W(\varrho, U)$ and satisfying any of the conditions (C1), (C2), or (C3) discussed above, is compatible with classical causal models and thus does not violate any classical inequality.
\end{prop}
\begin{proof}
The fact that condition (C1) implies classicality was already established in Proposition~\ref{prop:QM}. Next, invoking the quantum analogue of the mapping in Eq.~\eqref{eq:postselect-bell},
\begin{equation}
\Pr(AB\mid X)=\tr[(E_{A\mid X}\otimes G_{B\mid A})\,\varrho],
\end{equation}
it follows directly that the joint measurability of either $E_{A\mid X}$ or $G_{B\mid A}$ entails classicality. This likewise yields the same conclusion for CC-type process operators, for which $G_{B\mid A}$ is independent of~$A$. Finally, the convex hull of DC-type process operators coincides with the set of processes possessing no quantum memory, which, by Proposition~\ref{prop:QM}, are classical. By convexity, it follows that (C3) also implies classicality.
\end{proof}

One family of interesting processes that allows us to explore these notions of non-classicality are those consisting of a superposition between CC and DC, where the unitary $U$ is the two-qubit $\alpha$-partial \textsc{swap} gate
\begin{equation}
U_{\rm PS}(\alpha)=\cos\left(\frac{\alpha}{2}\right)\mathrm{id}+\ii\sin\left(\frac{\alpha}{2}\right)U_{\textsc{swap}}.
\end{equation}
For concreteness, we consider the state $\varrho$ to be the two-qubit Bell state, $\varrho=\ket{\mathrm{id}}\bra{\mathrm{id}}$. It is easy to see that $\alpha=0$ it is a DC process, whereas for $\alpha=\pi$ it is a CC-type process. For intermediate value, it can be neither nor a convex combination thereof. This can be proved by violations of classical causal inequalities. 

One particularly interesting family of processes for exploring these notions of non-classicality consists of superpositions of CC and DC mechanisms, where the unitary $U$ is taken to be the two-qubit $\alpha$-partial \textsc{swap} gate,
\begin{equation}
U_{\rm PS}(\alpha)=\cos\left(\frac{\alpha}{2}\right)\mathrm{id}+\ii\sin\left(\frac{\alpha}{2}\right)U_{\textsc{swap}}.
\end{equation}
For concreteness, we take the state $\varrho$ to be the two-qubit Bell state $\varrho=\ket{\mathrm{id}}\bra{\mathrm{id}}$. It is straightforward to verify that for $\alpha=0$ the process is DC, while for $\alpha=\pi$ it becomes CC-type. For intermediate values, however, the process is neither DC nor CC, nor even a convex combination of the two. This fact can be demonstrated through violations of classical causal inequalities.

A routine calculation reveals that this process exhibits quantum violations, quantified by the function
\begin{equation}\label{eq:gammaalpha}
\Gamma_\alpha=\frac{3-\sin \alpha +\cos \alpha}{2}.
\end{equation}
Quantum violations occur whenever $\Gamma_\alpha<1$, which holds for $\pi/2<\alpha<\pi$ (see Fig.~\ref{fig:violationswmrsr}). This implies that throughout this interval the process displays \emph{all} of the non-classical features discussed above.

\begin{figure}
    \centering
    \includegraphics[width=0.55\linewidth]{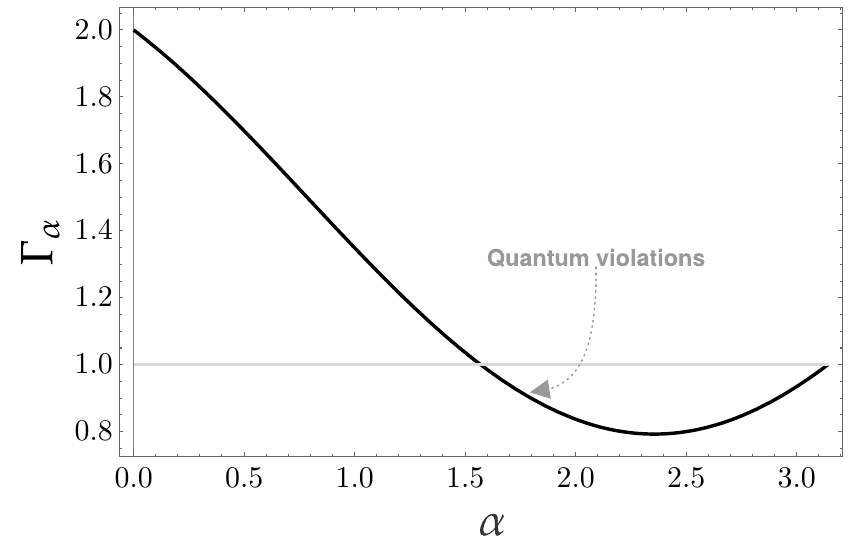}
    \caption{\sf\textbf{Quantum violations with partial $\alpha$-swap gate.} Plot of the $\Gamma_\alpha$ function [Eq.~\eqref{eq:gammaalpha}]. Quantum violations occurs when $\Gamma_\alpha<1$, which is the case if $\pi/2<\alpha<\pi$.}
    \label{fig:violationswmrsr}
\end{figure}

It is thus natural to inquire why quantum violations are confined to this particular interval. Indeed, Refs.~\cite{MRSR17,NQGMV21} showed that the process operator corresponding to $\alpha=\pi/2$ (termed there the ``coherent mixture of CC and DC'') cannot be expressed as an incoherent mixture of CC and DC, yet it nonetheless fails to violate the classical inequality. The reason quantum violations do not arise outside the interval identified above is that, although the process cannot be decomposed into a mixture of CC and DC, the $\alpha$-partial \textsc{swap} gate does not generate measurement incompatibility, as we demonstrate below.

\begin{prop}
Let $G_{B\mid A}^\alpha$ be the measurement assemblage generated through the $\alpha$-partial {\sc swap} gate, $0\leq \alpha\leq \pi$, via Eq.~\eqref{eq:GBASM}. If $A$ and $B$ are binary variables, $G_{B\mid A}^\alpha$ is always jointly measurable for $0\leq \alpha\leq \pi/2$ and $\alpha=\pi$ and can be incompatible otherwise.
\end{prop}
\begin{proof}
First, note that with some manipulations one can write
\begin{align}\label{eq:meas_assemb2}
G_{B\mid A}^{\alpha}=&\cos^2\left(\frac{\alpha}{2}\right)\tr(F_B\sigma_A)\mathrm{id}+\sin^2\left(\frac{\alpha}{2}\right)F_B-\ii \sin\left(\frac{\alpha}{2}\right)\cos\left(\frac{\alpha}{2}\right) [F_B,\sigma_A].
\end{align}
Joint measurability for $\alpha=\pi$ is trivial, since $F_{B\mid A}^{\hspace{0.025cm}\pi}=F_B$, i.e., it does not depend on $A$. Eq.~\eqref{eq:meas_assemb2} is linear in both states and POVM effects and thus we can restrict ourselves to projective measurements and pure states. Moreover, we explore the fact that $U_{\rm PS}(\alpha)$ lies within the symmetric subspace to set, without loss of generality, $\sigma_{A=0}=\ket{0}\bra{0}$ and $\sigma_{A=1}=\ket{\theta_{\rm S},\phi_{\rm S}}\bra{\theta_{\rm S},\phi_{\rm S}}$ as well as $E_{B=0}=\ket{\theta_{\rm E},\phi_{\rm E}}\bra{\theta_{\rm E},\phi_{\rm E}}$. Moreover, we parametrise the POVM effects as follows
\begin{equation}
F_{B=0|A}^{\hspace{0.05cm}\alpha}=\frac{1}{2}[(1+\gamma_A)\mathrm{id}+\boldsymbol{r}_A\cdot \boldsymbol{\sigma}],
\end{equation}
where $\boldsymbol{r}_A\cdot \boldsymbol{\sigma}=r_A^x \sigma_x+r_A^y \sigma_y+r_A^z \sigma_z$. We do so in order to apply the criterion for compatibility of a pair of qubit POVMs (see Ref.~\cite{Guhne2023} and references therein). The criterion states that $F_{B\mid A}^{\hspace{0.05cm}\alpha}$ is jointly measurable if and only if
\begin{equation}\label{eq:yu_criterion} %~\cite{Stano2008,Busch2009,Yu2010}
(\boldsymbol{r}_0\cdot \boldsymbol{r}_{1}-\gamma_0 \gamma_1)^2-\left(1-\mathcal{F}_0^2-\mathcal{F}_1^2\right)\left(1-\frac{\gamma_0^2}{\mathcal{F}_0^2}-\frac{\gamma_1^2}{\mathcal{F}_1^2}\right)\geq 0,
\end{equation}
where
\begin{equation}
\mathcal{F}_a=\frac{1}{2}\left(\sqrt{(1+\gamma_a)^2-\|\boldsymbol{r}_a\|^2}+\sqrt{(1-\gamma_a)^2-\|\boldsymbol{r}_a\|^2}\right).
\end{equation}
A routine calculation shows that for $A=0$ one has 
\begin{subeqnarray*}
\gamma_0&=&\cos^2\left(\frac{\alpha}{2}\right)\cos\theta_{\rm E}, \\
x_0&=&\sin\left(\frac{\alpha}{2}\right) \sin\theta_{\rm E}\sin\left(\frac{\alpha}{2}+\phi_{\rm E}\right) ,\\ 
y_0&=&-\sin\left(\frac{\alpha}{2}\right) \sin\theta_{\rm E}\cos\left(\frac{\alpha}{2}+\phi_{\rm E}\right), \\
z_0&=&\sin^2\left(\frac{\alpha}{2}\right)\cos\theta_{\rm E},
\end{subeqnarray*}
and, similarly, for $A=1$,
\begin{subeqnarray*}
\gamma_1&=&\cos^2\left(\frac{\alpha}{2}\right)\left[\cos\theta_{\rm E}\cos\theta_{\rm S}+\cos(\phi_{\rm S}-\phi_{\rm E})\sin\theta_{\rm E}\sin\theta_{\rm S} \right], \\
x_1&=&\sin \left(\frac{\alpha}{2}\right) \left[\cos\left(\frac{\alpha }{2}\right) (\sin \theta_{\rm E} \cos \theta_{\rm S} \sin \phi_{\rm E}-\cos \theta_{\rm E} \sin \theta_{\rm S} \sin \phi_{\rm S})+\sin \left(\frac{\alpha }{2}\right) \sin \theta_{\rm E} \cos \phi_{\rm E}\right], \\
y_1&=&\sin \left(\frac{\alpha }{2}\right) \left[\cos \left(\frac{\alpha }{2}\right) (\cos \theta_{\rm E} \sin \theta_{\rm S} \cos \phi_{\rm S}-\sin \theta_{\rm E} \cos \theta_{\rm S} \cos \phi_{\rm E})+\sin \left(\frac{\alpha }{2}\right) \sin \theta_{\rm E} \sin \phi_{\rm E}\right], \\ 
z_1&=&\sin\left(\frac{\alpha }{2}\right)\left[\sin\left(\frac{\alpha }{2}\right)\cos \theta_{\rm E}-\cos\left(\frac{\alpha }{2}\right) \sin (\theta_{\rm E}) \sin \theta_{\rm S} \sin (\phi_{\rm S}-\phi_{\rm E})\right].
\end{subeqnarray*}
It is easy to see that $\mathcal{F}_0=\mathcal{F}_1=\cos(\alpha/2)$. Plugging these expressions into Eq. (\ref{eq:yu_criterion}) we get
\begin{align}\label{eq:yu_criterion2}
    &f(\alpha,\theta_{\rm S},\theta_{\rm E},\phi_{\rm S},\phi_{\rm E})\cos\alpha+g(\alpha,\theta_{\rm S},\theta_{\rm E},\phi_{\rm S},\phi_{\rm E})^2\geq 0,
\end{align}
where $g(\alpha,\theta_{\rm S},\theta_{\rm E},\phi_{\rm S},\phi_{\rm E})=\boldsymbol{r}_0\cdot \boldsymbol{r}_{1}-\gamma_0 \gamma_1$ and
\begin{equation*}
f(\alpha,\theta_{\rm S},\theta_{\rm E},\phi_{\rm S},\phi_{\rm E})=1-\cos ^2\left(\frac{\alpha }{2}\right)\left\{\cos^2\theta_{\rm E}+\big[\sin \theta_{\rm E} \sin \theta_{\rm S} \cos (\phi_{\rm S}-\phi_{\rm E})+\cos\theta_{\rm E} \cos \theta_{\rm S}\big]^2\right\}.
\end{equation*}
Consequently, the inequality~\eqref{eq:yu_criterion2} cannot be violated for $0\leq \alpha\leq \pi/2$. 
\end{proof}

This result provides a rigorous explanation for why violations do not occur for values of 
$\alpha$ outside the interval $\pi/2<\alpha<\pi$. Yet, the argument applies only to the case of binary variables $A$ and $B$. Consequently, it remains an open question whether quantum violations might arise outside this interval when one considers non-binary variables.

\section*{D. Description of the experimental setup}
\label{app:setup}

We now briefly describe the experimental setup used in this work. The experiment was performed on the trapped-ion quantum processor reported in Refs.~\cite{Schindler2013,Ringbauer2022}, to which we refer for further technical details.

Fundamentally, we work with a string or chain of $^{40}\text{Ca}^+$ ions trapped in macroscopic linear Paul trap.
We manipulate the qubit states using an addressed \SI{729}{\nano\meter} laser with a linewidth below \SI{1}{\hertz}~\cite{freund_xcorr}, which illuminates the ion chain at an angle of \SI{67.5}{\degree}. This is so that it can still address individual ions while also coupling to the ions' axial motion. Figure~\ref{fig:trap} shows a schematic of the trap, including the propagation directions of the dipole and quadrupole lasers.

Single-ion addressing is achieved using a crossed acousto-optic deflector (AOD) setup~\cite{Ringbauer2022}, which allows us to address all ions in the chain individually. Due to a finite beam size and the angled incidence of this addressing beam, crosstalk between addressed ions still remains one of the biggest sources of infidelity in our setup. To mitigate this, we add a spacer ion in between the ions that implement the qubits we use. 
% \leo{This sounds a bit confusing. Do they allow to apply single qubit gates? What are those frequencies?} \peter{I rewrote it a bit and added two sentences, please check}

We perform state readout using the \SI{397}{\nano\meter} laser by resonantly scattering photons, which is collected with an objective and image onto an EMCCD camera. Photon scattering indicates population in the $\mathrm{S}_{1/2}$ manifold, whereas the absence of fluorescence corresponds to population in the $\mathrm{D}_{5/2}$ manifold. 

% The qubit states are manipulated using an addressed \SI{729}{\nano \meter} laser with a <\SI{1}{\hertz} linewidth~\cite{freund_xcorr}, and entanglement between ions is created with M\o lmer-S\o rensen gates~\cite{molmer_sorensen_1999,molmer_sorensen_2000} using the shared motional modes of the linear ion chain. In this experiment, only the centre-of-mass (COM) mode was used to mediate entanglement. A schematic of the trap with the direction of the dipole and quadrupole laser can be seen in Fig.~\ref{fig:trap}. State readout also utilises the \SI{397}{\nano \meter} laser, scattering light resonantly and collecting it using an objective and focusing it onto an EMCCD camera. Scattering light then corresponds to the ion being in the $\text{S}_{1/2}$ manifold, whereas no light to the $\text{D}_{5/2}$ manifold. Single-ion-addressing is achieved using a crossed-acousto-optical-deflector-setup~\cite{Ringbauer2022,aods_quantum_info}, with which individual ions can be addressed by selecting one or two frequencies which are applied to the AODs. 

In order to temporarily store information during mid-circuit measurement and re-cooling, we use additional levels in both the $\mathrm{S}_{1/2}$ and $\mathrm{D}_{1/2}$ manifolds, thereby enabling the mid-circuit operations necessary to implement our quantum memory certification protocol.

Entangling two-qubit gates in a chain of ions can be generated using Mølmer–Sørensen (MS) gates~\cite{molmer_sorensen_1999}. They rely on the ions' shared motional modes, which, in our experiments, are the centre-of-mass (COM) mode. More specifically, an MS-gate is typically implemented with two bi-chromatic laser pulses illuminating two ions simultaneously. The two frequencies of the bi-chromatic laser beams have to be $ \omega \pm (\nu + \delta)$, where $\omega$ is the qubit transition, $\nu$ the frequency of a motional eigenmode of the ion crystal and $\delta$ an additional detuning that can be used to tune the gate. During the duration of one gate, the motional mode executes a loop in its phase-space and ideally ends up exactly where it started (often the ground state).
Therefore, while perfect ground-state cooling is not required, the fidelity of a gate is still influenced by motional modes' temperature.

Mathematically, a MS-gate on ions $i$ and $j$ takes the form
\begin{equation}\label{eq:MSGATE}
\mathrm{MS}^{ij}_{\phi_i,\phi_j}(\pi/2) = \ee^{-\ii(\pi/2)\,\sigma_{\phi_i}\otimes\sigma_{\phi_j}}
\end{equation}
where $\phi_i$ is the relative phase of the laser pulses implementing the gate to the rest of the sequence they are in. This can be used to determine the axis of rotation of the gate via $\sigma_{\phi}=\cos\phi\,\sigma_x + \sin\phi\,\sigma_y$. Accordingly, the MS-gates can be used to implement both XX and YY interactions.

\begin{figure}\sf
    \begin{minipage}{0.45\linewidth}\flushleft
    \textbf{A.}\vspace{0.25cm}
    \includegraphics[width=1.1\linewidth]{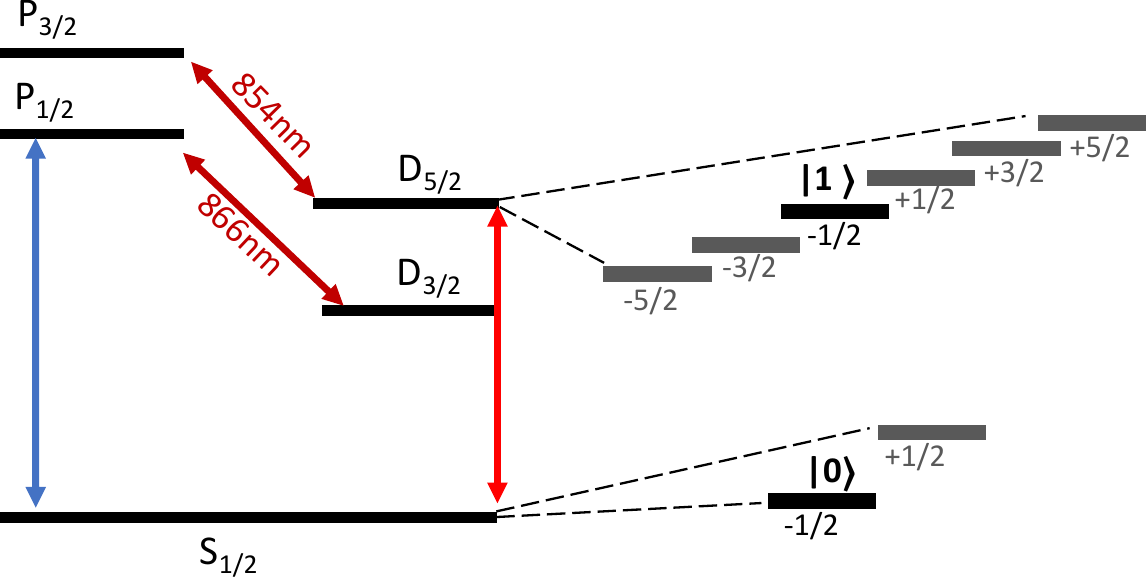}
    \end{minipage}
    \hspace{1cm}
    \begin{minipage}{0.45\linewidth}\flushleft
    \textbf{B.}\vspace{0.25cm}
    \includegraphics[width=1.1\linewidth]{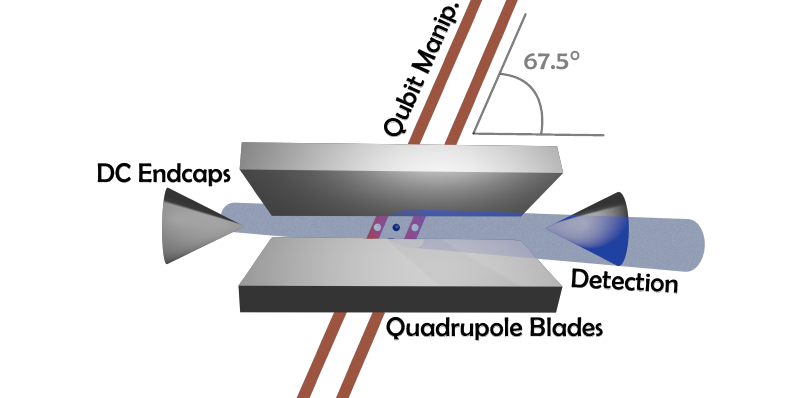}
    \end{minipage}
    \caption{\sf {\bf \textsf{A.}} Schematic of the level structure in the $^{40}\text{Ca}^+$ ion. The energy levels used for qubit states are labelled as $\ket{0}$ and $\ket{1}$. {\bf \textsf{B.}} Schematic of the ion trap and lasers. Both the dipole and the addressed quadrupole laser are coming in at an angle of \SI{67.5}{\degree} to the trap axis.}
    \label{fig:trap}
\end{figure}

\section*{E. Measurement sequences}
\label{app:sequence}

Starting with three ground-state-cooled ions in the trap, we prepare the two-qubit Bell state with a single M{\o}lmer-S{\o}rensen gate that we discussed in the previous section. 
% \peter{Is this reference necessary, the MS gate paragaph ends just the sentence prior? But if you think it's good we can leave it :)} \leo{now refer to the paragraph/equations added previously} 
on the outer ions. Afterwards, we measure qubit $\bf A^\prime$ locally; that is, ideally without influencing other ions. Since $A$ and $B$ are binary, we can consider the alphabet of the random variable $X$ having four distinct values, $\mathcal{X}=\{0,1,2,3\}$. They correspond, respectively, measurements of Pauli observables, $\{\sigma_x,\sigma_z,-\sigma_x,-\sigma_z\}$.

We implement selective mid-circuit measurement by shelving the qubit states of $\bf E$ to unused energy levels available in the $\text{D}_{5/2}$-manifold of the Calcium ion; see panel {\bf\textsf{A}.} of Fig.~\ref{fig:trap}. If $\bf A^\prime$ is found to be in the ground state (here denoted $\ket{0}$), it scatters light during the readout, heating up the ion's shared motional mode. This has to be re-cooled, using Doppler- and Polarisation Gradient Cooling before the circuit can continue. Resolved Sideband Cooling, necessary to reach the motional ground state again, is not possible in this setup without disturbing qubit $\bf E$. The leftover motional excitation constitutes the main contribution to the infidelity of all gates following the mid-circuit measurement.

After the readout, we re-prepare qubit $\bf A^\prime$ depending on the measurement outcome, and label it $\bf A$. Importantly, the decision-making is not implemented in real time, but through post-processing: individual experiment runs are assigned to variable $X$ only if the measurement outcome of $\bf A^\prime$ matches the re-prepared state. The rules for what state to prepare for which measurement outcome differ for the two different experiments, and can be found in Table~\ref{tab:post-processing}. Because of these specific choice of measurements no data has to be discarded, even though this assignment happens in post-processing. In fact, note that both the $\pm\sigma_x$ and $\pm\sigma_z$ are possible measurements and therefore, for every value of $X$ sampled, for every outcome $A$ observed and assigned post-measurement state, there exists a setting in which the sequence is correct. In other words, the association of any observed event $A=a$ to any of the two post-measurement state differ only in the different signs in the corresponding Pauli measurement (e.g., one corresponds to measuring, say $+\sigma_x$, while the other $-\sigma_x$); see Table~\ref{tab:post-processing}.

The readout is then followed by a unitary operation $U$ acting on both qubits, and a final measurement. We performed two different experiments. The first one implements a variable waiting time followed by a measurement aimed at certifying the quantum memory. The second is the partial $\alpha$-{\sc swap} gate, used to test aspects of non-classical causality as discussed above and in the main text. These two different experiments also differ in the state that is re-prepared after the first measurement, and the basis in which the second-time measurement is carried out.

Both experiments we describe suffer from correlated and uncorrelated bit-flip noise on the gates implementing the corresponding sequences. This predominantly stems from the unitary $U$, since, after the mid-circuit measurement, the ion chain cannot be ground-state cooled again without losing information. Accordingly, $U$ is performed on a \emph{heated} ion chain, resulting in a lower gate fidelity. Importantly, as this due to the measurement and not the specific gate implemented, it occurs in both experiments. 
% \leo{please check} \peter{Nice!}

In what follows, we shall describe the two protocols we implement in more detail. The measurements we performed, as well as the corresponding post-measurement states and their convention with the measurement outcomes, are summarised in the Table~\ref{tab:post-processing}.

\begin{table}[]
\centering
{\setlength{\tabcolsep}{0pt}
\begin{tabular}{ccc}
\hline
\rowcolor[HTML]{F5F5F5}
\quad\sf \textbf{Meas. basis} ($X=x$) \quad\quad & \quad\quad\sf \textbf{Outcome} ($A=a$) \quad\quad & \quad\quad\sf \textbf{Post-meas. state} $\rho_{A=a}$ \quad\quad \\ \hline
\rowcolor[HTML]{FFFFFF}
$\sigma_x$ ($X=2$) & $A=0$ & \sf\textcolor{midnightblue}{$\ket{+}$} \quad or \quad \textcolor{darkred}{$\ket{+\ii}$} \\
\rowcolor[HTML]{FFFFFF}
$\sigma_x$ ($X=0$) & $A=0$ & \sf\textcolor{midnightblue}{$\ket{-}$} \quad or \quad \textcolor{darkred}{$\ket{-\ii}$}\\
\rowcolor[HTML]{FFFFFF}
$\sigma_x$ ($X=0$) & $A=1$ & \sf\textcolor{midnightblue}{$\ket{+}$} \quad or \quad \textcolor{darkred}{$\ket{+\ii}$}\\
\rowcolor[HTML]{FFFFFF}
$\sigma_x$ ($X=2$) & $A=1$ & \sf\textcolor{midnightblue}{$\ket{-}$} \quad or \quad \textcolor{darkred}{$\ket{-\ii}$}\\
\rowcolor[HTML]{FFFFFF}
$\sigma_z$ ($X=3$) & $A=0$ & \sf\textcolor{midnightblue}{$\ket{+}$} \quad or \quad \textcolor{darkred}{$\ket{+\ii}$}\\
\rowcolor[HTML]{FFFFFF}
$\sigma_z$ ($X=1$) & $A=0$ & \sf\textcolor{midnightblue}{$\ket{-}$} \quad or \quad \textcolor{darkred}{$\ket{-\ii}$}\\
\rowcolor[HTML]{FFFFFF}
$\sigma_z$ ($X=1$) & $A=1$ & \sf\textcolor{midnightblue}{$\ket{+}$} \quad or \quad \textcolor{darkred}{$\ket{+\ii}$}\\
\rowcolor[HTML]{FFFFFF}
$\sigma_z$ ($X=3$) & $A=1$ & \sf\textcolor{midnightblue}{$\ket{-}$} \quad or \quad \textcolor{darkred}{$\ket{-\ii}$}\\ \hline
\end{tabular}}
\caption{\sf Post-processing assignment for the measurement performed at the first time. This table applies to both experiments, referred to as the “quantum memory test” and the “partial swap protocol.” The measurement basis is identical in the two experiments; they differ only in the assigned post-measurement states. In the quantum memory test these are \textcolor{midnightblue}{$\ket{\pm}$}, while in the partial swap protocol they are \textcolor{darkred}{$\ket{\pm \ii}$}.}
    \label{tab:post-processing}
\end{table}

\subsection{Quantum memory test}

In the quantum memory test we introduced a waiting time after the first measurement, followed by a {\sc cnot} and a full swap two-qubit gate between $\bf B$ and $\bf E$. For an ideal quantum memory, it theoretically produces the maximal quantum violation of inequality~\eqref{eq:KM-SM}. Yet, the experimentally measured violation decays with increasing waiting time, indicating a slow loss of quantum coherence of the information stored in the two systems. These data were presented in the panel {\bf\textsc{B.}} of Fig.~\ref{fig:experiment} in the main text; below, we explain how we implement the corresponding gates and measurements, including imperfections.

% The data, including a prediction of the memory decay from our quantum processor benchmarks, was shown in panel {\bf\textsc{B.}} of Fig.~\ref{fig:experiment} in the main text

First, the state re-preparation works as follows. If the measurement outcome is, say $A=1$, $\ket{+}$ is prepared, and for $A=0$ we prepare $\ket{-}$ (eigenvectors of $\sigma_x$). Recall that this does \emph{not} depend on the measurement basis, that is, the variable $X$. 

The {\sc cnot} gate is implemented as a combination of one MS-gate [Eq.~\eqref{eq:MSGATE}] and local rotations on control ($\bf C$) and target ($\bf T$) qubit in the form of 
\begin{equation}
U_{\textsc{cnot}} = R^{\bf C}_{0°}(\pi/2)\cdot R^{\bf T}_{\pi°}(\pi/2) \cdot R^{\bf T}_{\pi/4}(\pi) \cdot \mathrm{MS}^{\bf C,\,T}_{XX}(\pi/2)\cdot R^{\bf T}_{\pi/2}(\pi/2).
\end{equation}
Here, $R^{i}_{\phi}(\theta)$ is the rotation on the single-qubit Bloch sphere of qubit $i\in\{\mathbf{C},\,\mathbf{T}\}$ for a rotation angle of $\theta$ around the axis that lies on the equator and is rotated by $\phi$ with respect to the x-axis in the Bloch sphere. 

% The MS-gate in this case takes the form of $MS^{\bf C T}_{XX}(\pi/2) = e^{-i\frac{\pi}{2} \sigma^{\bf C}_X\otimes \sigma^{\bf T}_X}$. \leo{I think it is not necessary to change this as I said before; I just misread it. Yet, it is not defined which qubit is control and which one is target. I mean, replacing $0$ and $1$ by $\bf C$ and $\bf T$ makes it explicit and unambiguous.} 

We implement the {\sc swap} gate virtually; that is, instead of physically swapping the positions of ions in the chain or jointly addressing them in order to implement the unitary {\sc swap}, we simply change their roles. This is much more practical, and, from a theoretical standpoint, nothing changes since we explicitly treat qubits $\bf A^\prime$, $\bf A$ and $\bf B$ as potentially different physical systems all the times. 

% \peter{sounds good!}.
% \leo{I moved the part of the heating to the top, as it refers to both experiments not to this specific one.}

In order to compensate coherent phase noise caused by environment magnetic field oscillations, we use a dynamical decoupling scheme with echo pulses at intervals of \SI{2.5}{\milli \second}. This leads to the decay being dominated by incoherent dephasing, with a coherence time of $\text{T}_2 \approx \SI{364}{\milli \second}$. An additional loss of information occurs due to the finite infidelity of the pulses performing the echo sequence, which was characterised to be \SI{0.005}{}. The spontaneous decay of the qubit excitation $\text{T}_1 = \SI{1.17}{ \second}$  plays a secondary role at the \SI{}{ms}-timescales for which we can certify memory. A longer discussion of the dynamical decoupling scheme used is in Section F of this Supplementary information. 

Finally, recall that the first measurement heats up the setup, which is the main reason why we observe an initial offset from the maximal violation (see panel {\bf\textsf{B.}} Fig.~\ref{fig:experiment}). 

% \cite{steane_lifetime}

% \leo{Please check }

\subsection{Partial swap protocol}

% Similar to the classical case, also in a quantum mechanical formalism causal relations can be described as being of either direct-cause (DC) or common-cause (CC) nature. Local operations on qubit $\bf A$ in between the first and the second measurement can be understood as affecting solely the direct-cause influence of the first measurement on the second. A common-cause component can be added by e.g. swapping the qubits $\bf A$ and $\bf E$ in between the measurements. 
% A more general relation consisting of a mixture of DC and CC can be implemented using a partial swap gate $\mathbf{U}_{swap}(\alpha) = \text{cos}(\alpha/2) \mathbb{I} + i \text{sin}(\alpha/2) \mathbf{U}_{swap}$, where $\mathbf{U}_{swap}$ just exchanges the states of systems $\mathbf{A^\prime}$ and $\mathbf{B}$.
% The swap angle $\alpha$ then quantifies the CC and DC component in the relationship between the two temporally separated measurements. This is expected to only result in a violation of the classical causality constraints in Equation~\eqref{eq:km} in the main text if the swap angle lies in a regime between $\pi /2$ and $\pi$. Importantly, having only DC behaviour (for swap angle $\alpha=0$) results in no violation of classicality, while only CC (swap angle $\alpha  =\pi$) lies at the classical limit. \leo{Already discussed, so I think should be removed.}\leo{If no opposition, please delete this part.}

% \leo{I rewrote this part according to what we discuss; please check}\peter{I think it's good}

In the test of quantum causal relations, we implemented the $\alpha$-partial {\sc swap} gate
\begin{equation}
U_{\rm PS}(\alpha)= \cos(\alpha/2) \mathrm{id} + \ii\sin(\alpha/2) U_{\textsc{swap}}.
\end{equation}
In this experiment, not only the unitary $U$ is changed but also the post-measurement states and the measurement performed at the second time, in order to observe the maximum possible violation for the corresponding process (see Section~C for the theoretical analysis). Now, instead of the eigenstates of $\sigma_x$, the system $\bf A$ is prepared in one of the eigenstates of $\sigma_y$, denoted $\ket{\pm \ii}$, namely $\ket{+\ii}$ if $A=0$ and $\ket{-\ii}$ otherwise. The second time measurement is now the X-Pauli, $\sigma_x$.

To implement the $\alpha$-partial {\sc swap} gate in the trapped ion processor, we use the following combination of MS-gates
\begin{equation}
U_{\rm PS}(\alpha)=\mathrm{MS}_{\rm ZZ}(-\alpha/2) \cdot \mathrm{MS}_{\rm YY}(-\alpha/2) \cdot \mathrm{MS}_{\rm XX}(-\alpha/2).
\end{equation}
These three MS-gates used above differ only from a local rotation; namely, the $\mathrm{MS}_{\rm XX}$ and $\mathrm{MS}_{\rm YY}$ differ only in a $\pi/2$ phase shift, and $\mathrm{MS}_{\rm ZZ}$ can be implemented using $\mathrm{MS}_{\rm XX}$ and two Hadamard gates that diagonalise the gate in the computational basis. The angle by which the MS gates rotate can be controlled with the laser intensity, while gate time itself is kept fixed.

The expected theoretical behaviour is given by the $\Gamma_\alpha$ function in Eq.~\eqref{eq:gammaalpha}; see also Fig.~\ref{fig:CC-DC}. The experimentally observed value are shown in the panel {\bf\textsf{C.}} of Fig.~\ref{fig:experiment} (up to rescaling when compared to Fig.~\ref{fig:CC-DC}). There, one can clearly see that the measured violations match the sinusoidal shape and phase of the theoretical prediction.

\section*{F. Echo sequence and dephasing}
\label{app:dephasing}

% In many laboratory environments, magnetic field noise is not just incoherent noise stemming from an unidentifiable environment, but often includes coherent noise at certain frequencies. Particularly hard to control are noise components at the frequency of the of the mains electricity supply, or multiples thereof. These noise components, which in Europe are \SI{50}{\hertz} and multiples, cannot always be removed, and are therefore dealt with it using active field compensation. This way we are able to remove coherent noise components at \SI{50}{\hertz}, but not at \SI{100}{\hertz}.
% These can in turn be mitigated using dynamical decoupling. We use a  variable number of echo pulses after each \SI{2.5}{\milli \second} of wait time. The corresponding filter function minimizes the sensitivity to coherent noise at \SI{100}{\hertz} as well as DC. This comes at the cost of being sensitive to \SI{50}{\hertz}, but since they are already suppressed by the active field compensation, this has no problematic effect. 

In this section we discuss the noises behind the classicalisation of our quantum memory we observe when introducing the waiting time.

\begin{figure}
    \centering
    \includegraphics[width=0.75\linewidth]{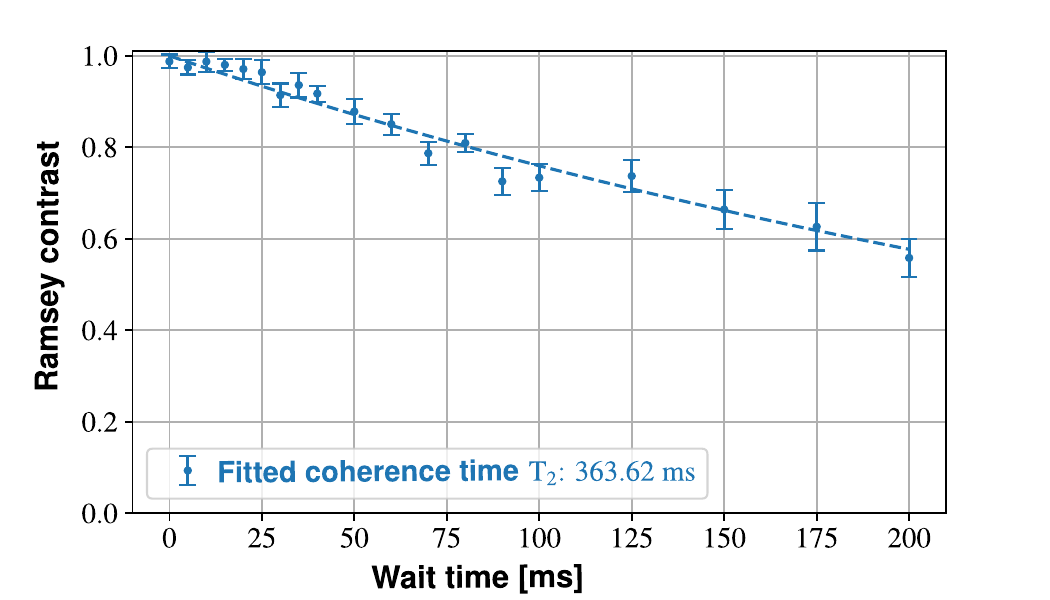}
    \caption{\sf \textbf{Coherence time measurement.} Measured dephasing time or effective coherence time $\mathrm{T}_2$ of the trapped ion processor. Here, ``Ramsey contrast'' refers to the result of a standard Ramsey experiment. It informs how much coherence is left. }
    \label{fig:ramsey}
\end{figure}

First, recall that, in many laboratory environments, magnetic-field noise is not purely incoherent noise arising from an uncharacterised environment, but often it also contains \emph{coherent} components at well-defined frequencies. Noise at the frequency of the mains electricity supply, and at its harmonics, is particularly difficult to control. In Europe, these components occur at \SI{50}{\hertz} and its multiples. Such noise cannot always be fully eliminated; instead, they can be addressed using active magnetic-field compensation. By doing so, we were able to suppress coherent noise at \SI{50}{\hertz}, but, unfortunately, we could not do the same for \SI{100}{\hertz}. 

To address this issue, we mitigate the remaining \SI{100}{\hertz} component using dynamical decoupling~\cite{lorenza_dd}. Specifically, we apply a variable number of echo pulses after each \SI{2.5}{\milli\second} waiting interval. The resulting filter function minimises sensitivity to coherent noise at \SI{100}{\hertz} as well as to quasi-static (DC) noise. Although this procedure increases sensitivity at \SI{50}{\hertz}, this does not pose a problem because those components are already suppressed by the active field compensation.

Finally, in order to simulate the expected decay rate, we measure the effective dephasing time of the dynamical decoupling sequence as well as the fidelity of the echo pulses used. 
% \leo{Would not it be dephasing time with dynamical decoupling?} \peter{I get a different dephasing time if I apply the DD sequence or not, which is why I write "effective dephasing time of the dynamical decoupling sequence". This is the relevant one, since we also apply the DD sequence during our memory certification. Is it kind of understandable what I mean?} 
Since always the exact same echo pulse is used in the gate sequence, calibration errors accumulate, yet we obtained an average fidelity of \SI{0.995}{} per echo. 

We measured the dephasing time with a set of Ramsey experiments (see e.g.~Ref.~\cite{Schindler2013}) over variable waiting time using an exponential fit. Of course, mathematically different models could be fitted to the data, yet the exponential fit is suitable for at least two reasons. First, the fit corresponds to frequency noise distributed like a Lorentzian around the centre frequency, which has a clear physical meaning and, second, as shown in Fig.~\ref{fig:ramsey}, it fits the observed data with great precision (specially in comparison with other possible models like Gaussian). With this data at hand, we simulate the expected decay of our quantum memory in Fig.~\ref{fig:experiment} of the main text. The agreement between our theory supplied with simple noise model and the observed decay shows the efficiency of our method to detect memory classicalisation.

\section*{F. Estimating cross-talk influences}
\label{app:crosstalk}

% \leo{I have added some text bellow; please check}

Finally, we discuss experimental inaccuracies that can lead to cross-talk influences. Since the absence of such causal influence is a central hypothesis which cannot be justified upon strong physical assumptions such as finite velocity of physical interactions (Einstein's relativity) which is often invoked in Bell's tests. In fact, as stressed in the main text, this assumption must be rigorously tested, which, as we have seen, can be done using ACDE and Pearl's inequality [Eq.~\eqref{eq:ACDE}~and~\eqref{eq:Pearl-ineq-SM}]. We will therefore present the main physical reasons behind possible imperfections in our setup and provide some data that are used to estimate such detrimental causal influence.

\begin{figure}
    \includegraphics[width=0.7\linewidth]{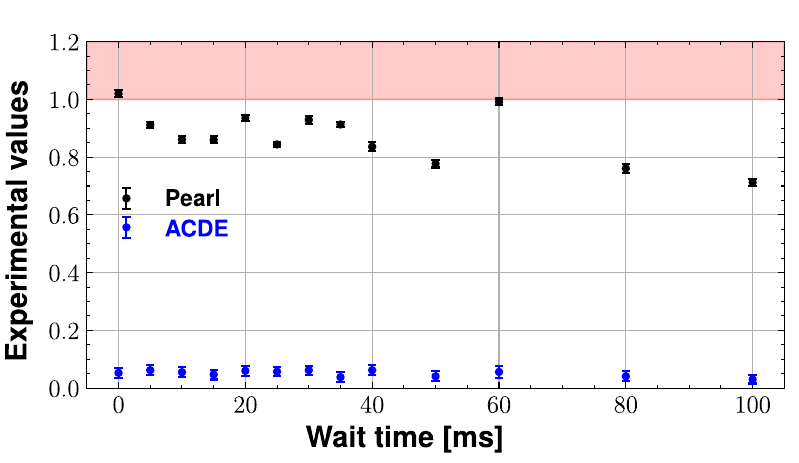}
    \caption{\sf \textbf{Pearl's inequality and the ACDE values over waiting time.} The threshold at which indicates cross-talk influences is given by the red-shaded area, corresponding to $\text{I}_{Pearl}>1$. The ACDE-values can be used to test Eq.~\eqref{eq:kmct} from the main text, and are also found to present no problem for the certification of our quantum memory.  }
    \label{fig:I_pearl}
\end{figure}

In our experiment, various error channels connecting the two systems $\bf A$ and $\bf E$ can add new causal influences on the measurements undertaken, opening a loophole in our certification scheme. On a physical layer, the main cause for remaining crosstalk-influences is what is called crosstalk of the laser beams addressing the individual ions. Since the ions are only spaced several \SI{}{\micro \meter} from each other, laser beams have to be very tightly focussed to only address one ion at a time. Still, its finite size, limited by its wavelength, as well as imperfections in the beam shape (e.g. Airy disks) result in a residual intensity of light at the ions supposed to be left dark. The tilt with which we address the ion chain (67.5 \degree) amplifies this effect, since it increases the effective beam cross-section in the ion plain by nearly \SI{10}{\%}. Crosstalk is strongly mitigated by using a spacer ion in between the two implementing the protocol, pushing the logical error per gate below \SI{e-3}{}. Still, we have to consider the effects it has on our certification scheme, whence the choice for ``cross-talk influences'' in our theoretical analysis.
% \peter{This is the crosstalk paragraph, please check}

% \leo{As discussed, we need a paragraph explaining physical cross-talk between ions.}\leo{Just as a note, the influence $X\to B$ does not invalidate DI but opens a loophole.}

At this point, some comments on device independence and associated loopholes are in order. First, as discussed in Section~A, Pearl’s inequality is formulated in terms of observational probabilities $\Pr(AB\mid X)$ and provides a \emph{sufficient} condition for the presence of cross-talk influences; that is, if the inequality is not violated, the existence of crosstalk cannot be ruled out. The ACDE, by contrast, yields a direct estimate of cross-talk influences, but it relies on interventional probabilities $\Pr(B_{\dop(A,\,X)})$. 

In the specific implementation of the experiment (see Section~E of this Supplementary Information), the same empirical probability table can be used to estimate both quantities. However, we emphasise that this does \emph{not} close the corresponding loophole, since doing so breaks device-independence by attributing two distinct causal structures, related through intervention, to the same data set. Accordingly, the measured data should be understood as \emph{addressing}, rather than \emph{closing}, the aforementioned loophole.

% In particular, we need to bound the causal influence of the instrumental variable $\bf X$ on the measurement of $\bf B$, not mediated by $\bf A$. One advantage of our protocol is that this can be done with the same data as the certification itself, using Pearl's inequality in the form of Eq. ~\eqref{eq:Pearl-ineq-SM} or the ACDE introduced in Eq.~\eqref{eq:ACDE} of the main text.
% The value of this inequality for the different wait times can be seen in Fig.~\ref{fig:I_pearl}. Both an $I_{Pearl}$-value close but below 1 and a $ACDE$-value close to 0 bound the effect of crosstalk influence. In the latter case,  Equation ~\eqref{eq:kmct} from the main text provides a bound on the effect of possible cross-causalities. Evaluating it with the ACDE values measured, they are insufficient to explain the measured violations up to wait-times of \SI{30}{\milli \second}.

For the quantum memory test, we estimated both Pearl's inequality and ACDE. The results with increasing waiting time are plotted in Fig.~\ref{fig:I_pearl}. For the case of an ideal quantum memory, the value we expect for Pearl's inequality to be
\begin{equation}
\frac{2+\sqrt{2}}{4}\approx 0.853.
\end{equation}
This value is very close to the optimally measured, $\approx 0.883$ [Eq.~(8) in the main text]. Its behaviour while increasing the waiting time in shown in Fig.~\ref{fig:I_pearl}, where one can also find the corresponding values for ACDE.

\end{document}